\documentclass[showpacs,amsmath,amssymb,twocolumn,prx,superscriptaddress,notitlepage]{revtex4-2}
\usepackage{lucas}
\usepackage{tikz-cd}
\usepackage{array}
\usepackage{stmaryrd}

\creflabelformat{equation}{#2\textup{(\textcolor{red}{#1})}#3}
\creflabelformat{figure}{#2\textcolor{red}{#1}#3}
\Crefname{equation}{Equation}{Equations.}
\Crefname{figure}{Figure}{Figures}

\begin{document}
\title{Quantum LDPC Codes for Modular Architectures}

\author{Armands Strikis}
\email{armands.strikis@mansfield.ox.ac.uk}
\affiliation{Department of Materials, University of Oxford, Oxford OX1 3PH, United Kingdom}
\affiliation{Quantum Motion, 9 Sterling Way, London N7 9HJ, United Kingdom}

\author{Lucas Berent}
\email{lucas.berent@tum.de}
\affiliation{Chair for Design Automation, Technical University of Munich, Germany}

\begin{abstract}
    In efforts to scale the size of quantum computers, modularity plays a central role across most quantum computing technologies. In the light of fault tolerance, this necessitates designing quantum error-correcting codes that are compatible with the connectivity arising from the architectural layouts. 
    In this paper, we aim to bridge this gap by giving a novel way to view and construct quantum LDPC codes \emph{tailored} for modular architectures.
    We demonstrate that if the intra- and inter-modular qubit connectivity can be viewed as corresponding to some classical or quantum LDPC codes, then their hypergraph product code fully respects the architectural connectivity constraints.
    Finally, we show that relaxed connectivity constraints that allow \emph{twists} of connections between modules pave a way to construct codes with better parameters. 
\end{abstract}

\maketitle

\section{Introduction}
In classical computing it has become standard to design architectures that divide the necessary processing power into smaller components instead of only increasing the power of a single system~\cite{buyya1999high,singh2015survey}. A similar trend can be observed in recent proposals around scaling quantum computation. A multitude of quantum computing platforms have natural limitations, e.g., on how many qubits may be contained within a single ion trap or a superconducting chip, whereas each instance of such a platform is referred to as \emph{a module}~\cite{Monroe2014, Bartolucci2021, gold2021, hucul2015, Buonacorsi_2019}. Scaling up existing systems is the main hurdle in current research. 
Therefore, \emph{modular architectures} that consist of many similar modules will likely be necessary~\cite{bravyi2022future, brown2022modular,bombin2021interleaving,ramette2023fault,niu2023low,bartolucci2023fusion}.

To execute large scale quantum algorithms, fault-tolerant quantum computation (FTQC) is essential~\cite{Preskill1997}. A crucial component of FTQC is the error-correcting code, which describes how to encode quantum information in a redundant way with the goal of lowering the error rates of computation~\cite{lidar_quantum_2013, Terhal2015}. Recent results have shown the existence of quantum low-density parity-check (QLDPC) codes with asymptotically good parameters~\cite{panteleev2022asymptotically,dinur2022good,DBLP:journals/tit/BreuckmannE21, leverrier2022quantum, lin2022good}. This is a strong indication that QLDPC codes may play a key role in lowering the qubit overhead necessary for FTQC. It is known that well performing QLDPC codes require ``long-range" qubit connectivity if it is desired to embed the system in some finite dimensional Euclidean space~\cite{tremblay2022constant}. In fact, the asymptotic scaling of code parameters is upper bounded by the scaling of long-range qubit connectivity~\cite{baspin2022connectivity,Baspin2022quantifying}. 

When specifically considering the practical setting of modular architectures of a quantum computer (with a finite number of qubits), we may expect that some degree of long range interactions that scale with the code size is physically feasible~\cite{Gambetta2020Rydberg, Landsman2019ion, HAFFNER2008,Choe2022kam}. These can be long-range interactions within each module or between the modules themselves. 
Because of this less constrained connectivity, the question whether it is useful and practical to favour QLDPC codes of finite size over the surface code---the current gold standard for many quantum computing platforms~\cite{Dennis2002Topo, KITAEV2003,bravyi1998quantum,krinner2022realizing, andersen2020repeated,marques2022logical,chen2021exponential,acharya2022suppressing}---is important. To answer this question, a multitude of aspects of FTQC need to be considered, such as the implementation of fault-tolerant logic, decoding, and the code performance. Most of these questions are still open for QLDPC codes. Previous works have mentioned the compatibility of QLDPC codes and modular architectures, but without providing the exact details for the code construction or the partition of the qubits into modules~\cite{campbell2017roads}. Alternatively, past works have in detail described the way to use a surface code for modular architectures~\cite{nickerson2013}, but do not consider general QLDPC codes.

In this paper, we explore aspects around fault-tolerance for modular architectures with a focus on code constructions.
In a formal and general way, we show how QLDPC codes \emph{tailored} to modular architecture connectivity constraints can be constructed.
This gives a correspondence between product constructions of QLDPC codes and modular architectures by viewing the intra- and inter-modular connectivities as Tanner graphs or equivalently as chain complexes of classical or quantum LDPC codes. First, we give a formal perspective on the recently introduced \emph{looped pipeline architecture}~\cite{cai2022looped}. We prove that it can be slightly extended to produce a 3D surface code. This immediately shows a valuable contribution of our formalism to practically highly relevant work around modular architectures. We then extend the construction to more general hypergraph product codes.
Finally, we show that a broader class of product codes with potentially better code parameters can be constructed for architectures that allow \emph{twisted} modular connectivity.
With this work, we take a further step towards closing the gap between physical ``low-level'' system architecture questions and recent theoretical breakthroughs around asymptotically good quantum codes~\cite{breuckmann2021quantum, panteleev2022asymptotically, leverrier2022quantum, dinur2022good,lin2022good}. Furthermore, we want to emphasize the need for investigations around practical applications of general QLDPC codes~\cite{panteleev2021degenerate}.

The rest of this work is structured as follows. First, modular architectures and a formal viewpoint is given in~\cref{sec:modules}. Notation and fundamental background is presented in~\cref{sec:chainC}. The looped pipeline architecture, the 3D surface code construction, and our firsts main theorem is discussed in~\cref{sec:looped-arch}. Stepwise generalisations are subsequently given in~\cref{sec:hgp} where first the intra-modular connectivity is generalised, then \cref{sec:beyond} where the inter-modular connectivity is generalised, and~\cref{sec:bp} where the allowed connections between the modules are generalised. Finally, we conclude with a short discussion and outlook in~\cref{sec:conclusion}.

\section{Modular Architectures}\label{sec:modules}
\begin{figure}
    \centering
    \includegraphics[width=\linewidth]{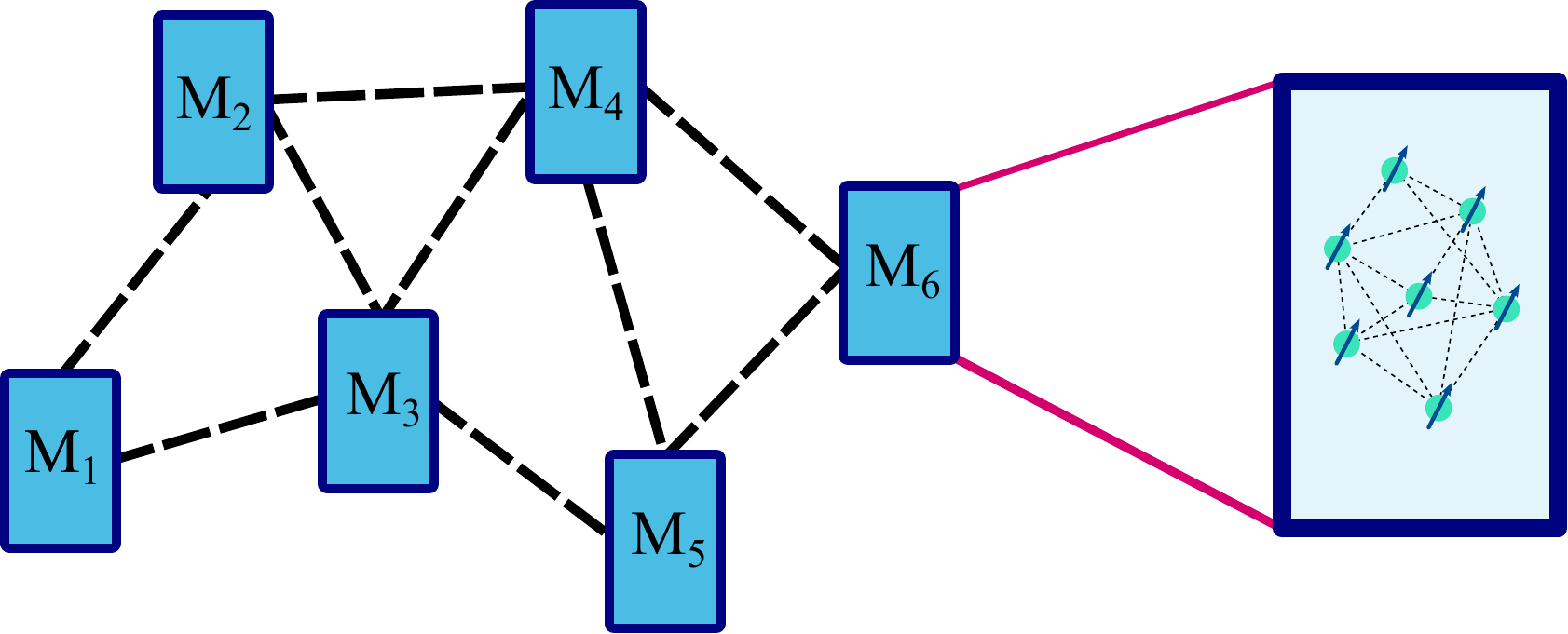}
    \caption{Here, quantum computing architectures where modules $\{M_k\}$ have a defined sparse inter-modular connectivity are considered. Each module contains a finite and equal number of qubits with the same connectivity constraints.}
    \label{ModArchitecture}
\end{figure}

In this paper we consider various qubit connectivity constraints that may arise in a quantum computer based on a modular architecture. Taking the constraints into account, we provide a novel recipe on how to construct quantum LDPC codes that respect a certain modular architecture.

Let us first concretely define what we mean by a \emph{modular architecture}. Consider a quantum computer with a (finite) collection of qubits $\{q_N\}$. In a modular architecture, each qubit is assigned to one and only one module, where the (finite) collection of such modules is given as $\{M_k\}$, see~\cref{ModArchitecture}. We assume that the modules are equivalent copies of each other. Therefore, we can partition the collection of qubits into disjoint sets $\{q_i\}_k$ such that $\{q_N\} = \bigcup_k \{q_i\}_k$, and we define that each module contains a finite and equal number of qubits $n$, i.e. $|\{q_i\}_k| = n$ for all $k$. To simplify the notation, we use the same canonically ordered index set for each module and hence drop the subscript $k$ altogether.
With this in mind, we can define the intra-modular qubit connectivity in the usual sense as follows.

\begin{definition}
\label{def:qubits}
A qubit $q_i \in M_k$ is \textit{connected} to a qubit $q_j \in M_k$ if the architecture allows us to directly implement two-qubit entangling operations between these qubits for all $k$. 
\end{definition}
\noindent Generally, we consider entangling gates such as controlled-not (CNOT) which are required for most syndrome circuits. However, entangling operations using measurements, e.g., in using photonic links, or otherwise are just as valid. We only require that these gates allow the construction of a syndrome extraction circuit required for quantum error correction.

In this work, we investigate cases where the qubit intra-modular interactions are given by a connectivity graph which can be viewed as a 
\emph{Tanner graph} of some classical or quantum code.

A (quantum) Tanner graph is a graph which has edges between nodes representing parity checks and data (qu)bits if and only if the (qu)bit is in the support of the parity check. See~\cref{TannerGraph} for a Tanner graph of a $7$ qubit Steane code and~\cref{sec:chainC} for a more technical introduction to Tanner graphs. As an example, nearest neighbour connectivity for a 1D chain of qubits corresponds to a Tanner graph of a classical repetition code.

In a modular architecture, each module may be connected to some number of other modules in a specific way.
We define the inter-modular connectivity as follows.

\begin{figure}
    \centering
    \includegraphics[width=0.9 \linewidth]{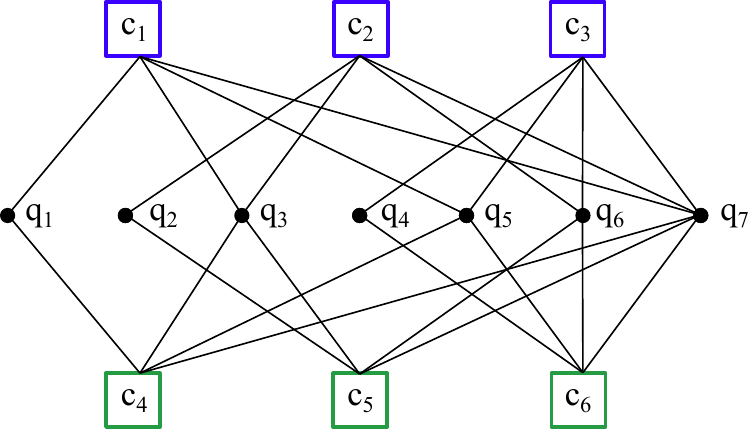}
    \caption{Tripartite quantum Tanner graph of a $7$ qubit Steane code. $\{q_i\}$ denote the data qubits, while $\{c_j\}$ denote the X parity checks (blue) and Z parity checks (green). Edges are drawn between them iff the data qubit is in the support of the check.}
    \label{TannerGraph}
\end{figure}

\begin{definition}
\label{def:modules}
A module $M_k$ is \textit{connected} to a module $M_j$ if the architecture allows us to directly implement two-qubit entangling operations between a qubit $q_i \in M_k$ and its \textit{respective} qubit $q_i \in M_j$ for all $i$.

\end{definition}
\noindent In~\cref{sec:bp-chain-compl}, this requirement will be slightly generalised to allow for \textit{twists} of the connections between modules to construct better quantum codes~\cref{def:modules2}.
Similarly to the intra-modular connectivity case, we choose the graph defining the inter-modular connectivity to correspond to a Tanner graph of a potentially different quantum or classical code. Finally, we say that the code \emph{respects} the connectivity constraints if we can associate a physical qubit in our quantum system to every parity check and data qubit of the code such that the parity check qubits are connected to the data qubits in their support.

In the following, we describe a way to create new codes that respect the overall architectural connectivity constraints as defined above. We only require that the intra- and inter-modular connectivities are formulated as Tanner graphs of some codes. To keep this formulation as general as possible, we employ the language of chain complexes.

\section{Preliminaries}\label{sec:chainC}
In this section, we discuss quantum codes and how to view them in terms of a $(\F_2)$ homological perspective~\cite{bombin2006topological, bravyi1998quantum, freedman2002z2, bravyi1998quantum, breuckmann2021quantum}. 

\subsection{Classical and Quantum Codes}
Since quantum LDPC codes have an interesting correspondence to classical codes, let us briefly discuss their classical analogue: binary linear codes. 

A classical binary linear $[n,k]$-code is a subspace $\calC \subseteq \F_{2}^{n}$. The set of codewords is called the \emph{codespace} and corresponds to the $k$-dimensional subspace $\ker H$ of a binary matrix $H$ called the \emph{parity check matrix} (pcm):

\[\calC = \set{x \in \F_2^n \mid Hx = 0}.\]

\noindent It is useful to describe code $\calC$ with its \emph{Tanner graph}. This is a bipartite graph $\calT(\calC)$ whose adjacency matrix is $H$.

Since stabilizer codes~\cite{gottesman1997stabilizer} play a central role in QEC, let us recall some fundamental definitions. 
We consider an $n$-qubit Hilbert space $(\C^2)^{\otimes n} = \C^{2^{n}}$. %
Let $\calP_n$ denote the (non-abelian) group of $n$ qubit Pauli operators defined as 

\begin{equation*}
\calP_n = \angles{i, X_j, Z_j \mid j\in [n]} =
\biggl\{ \phi \bigotimes_{j=0}^n P_j\biggr\},	
\end{equation*}	
where $\phi\in \set{\pm 1, \pm i}$ and $P_j$ is a single qubit Pauli operator $P_j\in\set{I,X,Y,Z}$. The weight $wt(P)$ of a Pauli operator $P$ is the number of non-identity components in the tensor product representation of $P$. A \emph{stabilizer group} $\calS$ is an abelian subgroup of $\calP_n$ s.t. $-I \notin \calS$. Elements of $\calS$ are called \emph{stabilizers} and the group is generated by $m$ independent \emph{stabilizer generators} $\calS = \angles{S_1, \dots, S_m}$.

The main idea of the \emph{stabilizer codes} is to use a common $+1$ eigenspace of all elements of a stabilizer group $\calS \subset \calP_n$ as the code space of a code $\calC$.
Therefore, an $\llbracket n,k,d \rrbracket$-quantum stabilizer code $\calC$ is a $2^k$-dimensional subspace of $(\C^2)^{\otimes n}$. Parameter $d$ denotes the \emph{minimal distance} of $\calC$, given by the minimal weight of a Pauli operator that commutes with all stabilizers $S_i$ but is not in the stabilizer group $S$.

Each $n$ qubit Pauli operator can be written as a binary vector. More formally, the quotient group $\calP_n/\set{\pm I^{\otimes n}, \pm iI^{\otimes n}}$ is isomorphic (up to phases) to $\F_2^{2n}$ by the isomorphism that sends an $n$ qubit Pauli operator corresponding to a tensor product of $X$ and $Z$ Paulis to a binary vector representation $(x|z) \in \F_2^{2n}$.
Thus $P,Q \in \calP_n$ commute iff for their binary representations $P \cong (x|z), Q\cong (x'|z')$ it holds that.

\begin{equation}
 \angles{x,z'} + \angles{z, x'} = 0\label{eq:pauli-innerprod-commute}
\end{equation}
This representation can naturally also be applied to the $m$ stabilizer generators $S_1, \dots S_m$ of a code $\calC$, which yields a $m \times 2n$ matrix $H=(H_X \mid H_Z)$. Each row of $H$ corresponds the binary representation of a stabilizer generator $S_i$. As for classical codes, matrix $H$ is the \emph{parity-check matrix} of $\calC$. By~\cref{eq:pauli-innerprod-commute}, $\ker H$ is exactly the set of vectors $(z|x) \in \F_2^{2n}$ s.t. their reordered form $(x|z)$ is the binary representation of a Pauli operator that commutes with all stabilizer generators $S_1, \dots S_m$.

An important subclass of stabilizer codes are CSS codes, which is considered in this work if not stated otherwise. These are stabilizer codes where all non-identity components of stabilizer generators are either all $X$ or all $Z$. Hence, the commutativity relation (\cref{eq:pauli-innerprod-commute}) can be written as 

\begin{equation}
	H_XH_Z^T=0. \label{eq:css-commutativity}
\end{equation}
Or equivalently, $C_Z^{\perp} \subseteq C_X$. Since the rows of a pcm correspond to the \emph{checks} of the code, elements of $H_X$ and $H_Z$ are called $X$ and $Z$ checks, respectively.
A CSS code is called \emph{low-density parity check} (LDPC) code if all checks have constant weight and each qubit is involved in a constant number of checks, i.e., if the parity check matrix (matrices) are sparse.

Let us now introduce an alternative perspective on codes that was essential in recent results around asymptotically good quantum and locally testable classical codes~\cite{DBLP:journals/tit/BreuckmannE21, panteleev2022asymptotically, leverrier2022quantum, dinur2022good, lin2022good}, and has become standard. 

\subsection{Chain Complexes}\label{sec:chain-complexes}
A \emph{chain complex} of vector spaces is a collection of vector spaces $\{C_i\}$ together with linear maps $\bdry_i$
\[ \bdry_{i}: C_i \to C_{i-1},\]

\noindent with the condition that squared boundary maps vanish, i.e.,
 
\begin{equation}
 \bdry_{i} \bdry_{i+1} = 0. \label{eq:boundary-sq}
\end{equation}

\noindent \cref{eq:boundary-sq} is equivalent to requiring that $\text{im } \bdry_{i+1} \subseteq \text{ker } \bdry_i $. Elements in $C_i$ are called \emph{i-chains} and 

\begin{align}
	& Z_i(C) = \text{ker } \bdry_i \subset C_i \\
	& B_i(C) = \text{im } \bdry_{i+1} \subset C_{i} \\
	& H_i(C) = Z_i(C)/B_i(C)
\end{align}
are the $i$-\emph{cycles}, $i$-\emph{boundaries} and the $i$-th \emph{homology} of the complex $C$ respectively. For instance, when considering chain complexes arising from simplicial complexes, $2$-chains correspond to formal linear combinations of faces, $1$-chains to formal sums of edges and $0$-chains to formal sums of vertices. Intuitively, $1$-cycles are loops that start and end in the same vertex and boundaries are those cycles that are a boundary of a set of faces.

A classical binary linear code $\calC$ can be viewed as a two-term chain complex: 

\begin{equation}
	C = C_{1} \xrightarrow[]{\partial_{1}} C_{0},
\end{equation}
where $C_1 = \F^n_2$ and $\partial_i = H$ is the parity check matrix. Then the code $\calC{}$ is the space of $1$-cycles: 
\[\calC = Z_1(C) = \text{ker } \bdry_1, \]
and the space of $0$-chains is the space of checks acting on $\calC{}$. Note that $H_0(C) = 0$ if the checks are linearly independent. For classical codes, this representation does not yield any new insights and hence is rarely used. However, a quantum CSS code necessitates commutation relations between the bit-flip and phase-flip parity check matrices $H_X$ and $H_Z$ (\cref{eq:css-commutativity}). Thus, there is a bijection between CSS codes and chain complexes: A CSS code corresponds to a three term chain complex:

\begin{equation}
	C = C_{i+1} \xrightarrow[]{\partial_{i+1}} C_{i} \xrightarrow[]{\partial_{i}} C_{i-1},
\end{equation}

\noindent where $\partial_{i+1} = H_Z^T$ and $\partial_{i} = H_X$. Thus, qubits are associated with $1$-chains and $X$ and $Z$ checks with $0$-and $2$-chains, respectively. A prototypical example is a toric code, where $C_{2}$, $C_1$ and $C_{0}$ are vector spaces of \textit{faces}, \textit{edges} and \textit{vertices}, obtained from a square cellulation of a torus. Conversely, given an arbitrary chain complex
\begin{equation}
	... \xrightarrow[]{} C_{i+1} \xrightarrow[]{\partial_{i+1}} C_{i} \xrightarrow[]{\partial_{i}} C_{i-1} \xrightarrow[]{} ...
\end{equation}
we can pick a dimension $i$ to associate the space of qubits with and view the corresponding three term chain complex as a CSS code.
The code parameters are $n = \text{dim } C_i$, $k = \text{dim } H_i$ and $d$ is the minimum weight of a non-trivial representative of $H_i$.

To a three-term chain complex we can associate a quantum Tanner graph by considering a simple mapping. A quantum Tanner graph $G$ is a tripartite graph with a vertex set $V = P_Z \sqcup Q \sqcup P_X$, where the vertices $Q$ are associated with qubits and vertices $P_{X(Z)}$ are associated with $X(Z)$ parity checks. Then, for an arbitrary three-term chain complex $C$, we define a one-to-one mapping where the basis elements of the vector space $C_{i+1}$ are mapped to vertices in $P_Z$, the basis elements of $C_{i}$ to $Q$, and the basis elements of $C_{i-1}$ to $P_X$. Finally, the edges between vertex partitions $P_Z$ and $Q$ are given by $\partial_{i+1}$ such that an edge exist between vertices $p\in P_Z$ and $q\in Q$ if and only if $\overline q\in C_{i}$, which is the corresponding basis element of $q$, is in the span of $\partial_{i+1}\overline p$, where $\overline p\in C_{i+1}$ and corresponds to the vertex $p$. Edges between $P_X$ and $Q$ are defined in an analogous manner using $\partial_i$. Note that there are no edges between any two vertices within the same partition or between partitions $P_X$ and $P_Z$.
Graphically, if the chain complex is drawn as an object with faces, edges and vertices, then each individual face and edge is replaced by a vertex. The partition of vertices is naturally induced from this mapping and edges exist only between those pairs of vertices that correspond to adjacent faces and edges (or vertices and edges) in the original chain complex.

Conceptually the same mapping to a graph can be applied to an arbitrary length chain complex. For example, a two-term chain complex would be mapped to a classical Tanner graph. Because of this bijection we will refer to chain complexes, their boundaries or their corresponding Tanner graphs (sometimes called connectivity graphs) interchangeably for the rest of the paper.

\begin{figure*}
    \centering
    \includegraphics[width=\linewidth]{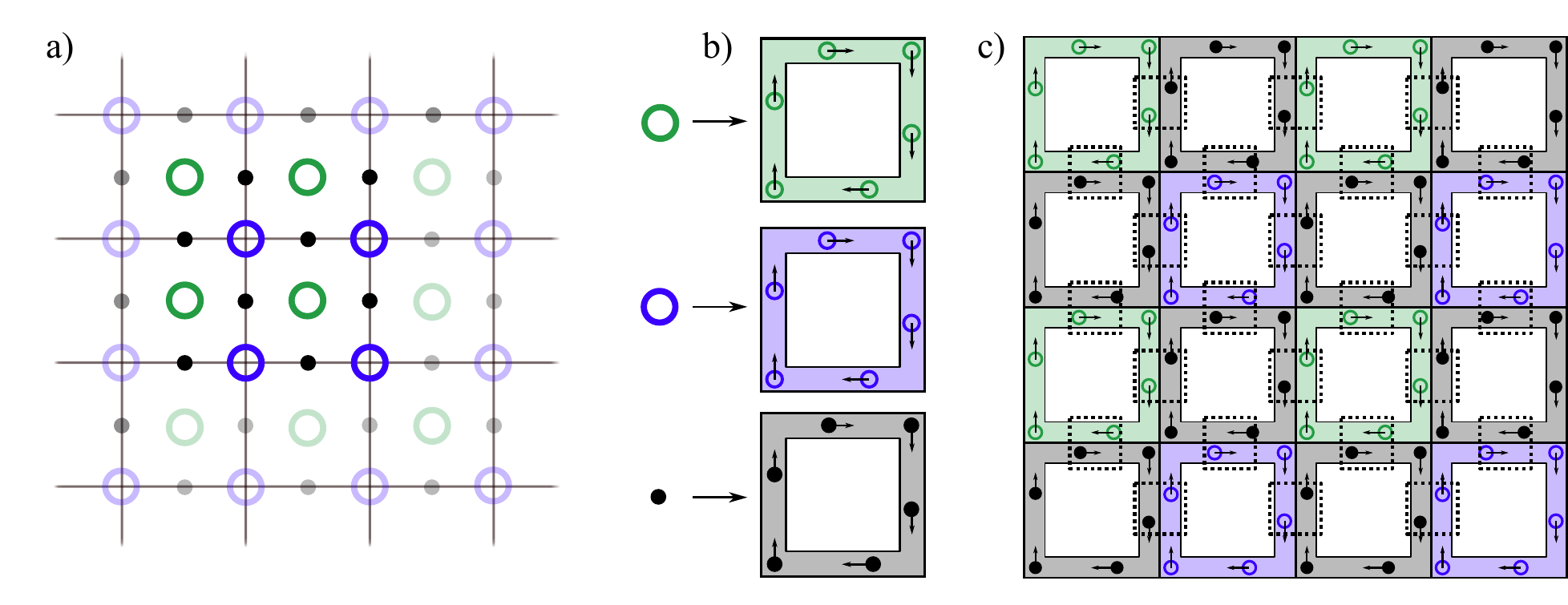}
    \caption{A stack of 2D surface codes can be generated by replacing every data qubit (black), Z parity check qubit (green) and X parity check qubit (blue) of the regular surface code (a) by a loop of corresponding qubits (b). The qubits communicate (entangle) with each other once they enter the dashed regions depicted in (c).}
    \label{Looped2D}
\end{figure*}

\section{Looped Pipeline Architecture}\label{sec:looped-arch}
To form a better intuition about our construction of codes for modular architectures, we first recap basic ideas of a recent result around QEC on a looped pipeline architecture~\cite{cai2022looped}. This constitutes the basis from which we build more involved and better quantum codes as we soften the constraints of the intra- and inter-modular connectivity and generalise the construction.

The work by Z. Cai and others considers a surface code layout (qubits are put on edges, $Z$-checks on faces and $X$-checks on vertices) as depicted in~\cref{Looped2D}(a) of a quantum chip where every data and ancilla qubit is replaced by a rectangular loop of fixed number of qubits as sketched in~\cref{Looped2D}(b). All types of qubits are moved along the loop in the clockwise direction at the same frequency. Once the qubits approach another qubit from a different loop, they interact to become entangled in a way that corresponds to the syndrome extraction circuit as illustrated in~\cref{Looped2D}(c). After the ancilla qubits have gone around the full loop, they are measured to read out the syndrome. The authors showed that this forms a stack of 2D surface codes, where the stack size is given by the number of qubits within each loop.

In their work, the authors and independently M. Fogarty~\cite{privcorrfogarty} identified that the stack of 2D surface codes may be used to generate a 3D surface code, but did not provide an explicit construction.
In the rest of the section we show how the looped pipeline architecture can be extended to implement a single 3D surface code (explained below) by assuming an additional connectivity within each qubit loop. Finally, we formalise and generalise this construction to the setting of modular architectures. This will allow us to construct quantum LDPC codes for more general connectivity constraints.

\subsection{3D Surface Code}
\label{sec:3Dsurface}

Consider a single data qubit loop from the looped pipeline architecture described above. Each qubit in this loop is part of a separate 2D surface code. If we extend the connectivity between the nearby qubits within each loop then we obtain a stack of surface codes linked together into a single block. This construction does not immediately produce a 3D surface code. For example, it links ancilla qubits to other ancilla qubits within the same loop.

\begin{figure}
    \centering
    \includegraphics[width=\linewidth]{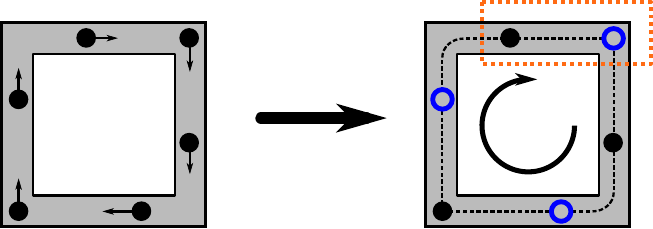}
    \caption{We replace all loops with a loop that has both ancilla (blue) and data (black) qubits. Additionally, we allow for nearby qubits in the loop to communicate. For example, the nearby qubits entangle whenever both of them enter the orange dashed box.}
    \label{RepetitionLoop}
\end{figure}

Instead, to produce a valid 3D surface code we additionally need to re-identify the qubits within each loop. In this regard, we form three different types of loops---face, edge and vertex loops---where the naming will parallel the chain complexes. They are laid out in a similar pattern as previously with face loops replacing the $Z$ ancilla loops, edge loops replacing the data qubit loops and vertex loops replacing the $X$ ancilla loops. Each of these loops may contain both data and ancilla qubits of the 3D surface code, hence, they also must contain measurement devices as described in~\cite{cai2022looped} to extract the syndrome~\cref{RepetitionLoop}. We assume that the total number of qubits per loop is even. Then, depending on the type of the loop and position within the loop each qubit can be given an assignment. These assignments can be found in~\cref{Tab1}. The even/odd parity of the qubit corresponds to its index $i$ within the loop, where the indexing is such that any qubit $q_i$ of any loop gets to interact with qubits $q_i$ of the neighbouring loops.
Note that we can exclude the even qubits in the face loop as they have no assignment.

\begin{table}
\begin{tabular}{|m{2.5cm} | m{2cm} m{2cm} |} 
 \hline
 Loop label & Odd qubits & Even qubits \\ [0.5ex] 
 \hline
 Vertex & X-stab ($6$) & Data \\ 
 
 Edge & Data & Z-stab ($4$)  \\
 
 Face & Z-stab ($4$) & --- \\ 
 \hline
\end{tabular}
\caption{Qubit assignments of different loops to generate a 3D surface code. Even/Odd assignment of qubits indicate their position in the qubit chain within the loop. Numbers in the parenthesis indicate the weight of the stabiliser.}
\label{Tab1}
\end{table}

By viewing the modular architecture in terms of chain complexes that describe codes, we can prove that such an assignment indeed produces a 3D surface code. In order to do so, we use tensor products of chain complexes.

\begin{remark}
The presented construction of a 3D surface code by shuttling the qubits around in loops will not have a threshold in general. Since the entangling and measurement operations are done on a one by one basis, the time it takes to extract syndrome is proportional to the number of qubits within each loop. A slightly altered scheme was proposed by~\cite{cai2022looped}, where the qubits in adjacent loops circulate in opposite directions and all qubits on the same side of the loop are being entangled at the same time with their respective qubits in other loops. If this was supplied with a number of measurement devices that scale in proportion with the number of qubits in the loop, the whole syndrome extraction would require $O(1)$ operations (e.g., one operation constitutes moving qubits to a new side of the loop). However, from a physical perspective one could argue that, in general, as the number of qubits per loop increases, so does the size of the loop that is needed. Therefore, the physical time it takes for a qubit to be shuttled at a finite speed to a new side of the loop (to perform a single operation) scales with the number of qubits per loop and, hence, the syndrome extraction time scales with the total number of qubits. In turn, errors accumulate faster than they can be corrected and the threshold for the code would not exist. This might be a general consequence for any non-concatenated quantum code for which the shuttling is used to do long range entangling gates locally (the same conclusion was reached by numerical analysis for a conceptually similar setup in~\cite{delfosse2021bounds}). Some ideas, like the ones presented in a recent work on hierarchical memories~\cite{pattison2023hierarchical}, may be used to overcome this challenge. Note that the aforementioned drawback only applies to this specific scheme based on shuttling. The main ideas in the rest of the paper do not share the same consequences.
\end{remark}

\subsection{Tensor Product of Chain Complexes}\label{sec:tensor-prod-chain-complexes}
Quantum codes can be constructed from products of chain complexes that describe other codes~\cite{tillich2013quantum, zeng2019higher}.
Let us discuss the construction in the following. The \emph{double complex} $C \boxtimes D$ is defined as

\begin{equation}\label{eqn:doubleComplex}
	(C \boxtimes D)_{p,q} = C_p \otimes D_q.
\end{equation}

	\begin{figure}	
	\begin{tikzcd}[cells={nodes={minimum height=2em}}]
		C_1 \otimes D_1 \arrow[r, "id^C \otimes \bdry^D_1"] \arrow[d, "\bdry^C_1 \otimes id^D"] & C_{1} \otimes D_{0} \arrow[d, "\bdry^C_1 \otimes id^D"]  \\
		C_0 \otimes D_1 \arrow[r, "id^C \otimes \bdry^D_1"] & C_{0} \otimes D_{0}
	\end{tikzcd}		 
\caption{A commuting diagram representing a double chain complex of two $2$-term chain complexes $C,D$.}
\label{fig:double-complex}
\end{figure}

\noindent with \emph{vertical} boundary maps $\partial_i^v = \partial_i^C \otimes \mathrm{id}^D$ and \emph{horizontal} boundary maps $\partial_i^h = \mathrm{id}^C \otimes \partial_i^D$ such that $\partial_{i}^v\partial_{i+1}^v = 0$, $\partial_{i}^h\partial_{i+1}^h = 0$, and $\partial_i^v\partial_j^h = \partial_j^h\partial_i^v$. An example double complex where $C,D$ are $2$-term complexes is visualized in~\cref{fig:double-complex}. The \emph{total complex} arises when we collect vector spaces of equal dimensions, i.e., ``summing over the diagonals'' in the double complex as follows

\begin{equation}
	\label{eq:TotalComplex}
	\mathrm{Tot}(C\boxtimes D)_n = \bigoplus_{p+q=n} C_p \otimes D_q = E_n
\end{equation}
where the boundary maps are $\partial^{E} = \partial^v \oplus \partial^h$.
Then, the \emph{tensor product complex} $C\otimes D$ is defined as $\mathrm{Tot}(C \boxtimes D)$. For the example given in~\cref{fig:double-complex}, the tensor product complex is

\[C_1 \otimes D_1 \xrightarrow[]{\partial_{2}} C_0 \otimes D_1 \oplus C_1 \otimes D_0 \xrightarrow[]{\partial_{1}} C_0 \otimes D_0,\]

where 
\begin{align*}
    &\bdry_2 = \left( \begin{matrix} \bdry^C_1 \otimes id^D \\ id^C \otimes \bdry^D_1 \end{matrix} \right)\\
    &\bdry_1 = \left(id^C \otimes \bdry^D_1 \mid \bdry^C_1 \otimes id^D\right).
\end{align*}

As the homology of a chain complex is related to the parameters of the corresponding code, the K\"unneth formula is central. It gives a method to compute the homology of a double complex, from the homology of the vertical and horizontal complexes:
\begin{equation}\label{eq:Kuenneth1}
    H_n(C \otimes D) \cong \bigoplus_{p+q=n} H_p(C) \otimes H_q(D).\nonumber
\end{equation}

\subsection{3D Surface Code in the Chain Complex Formalism}
\label{sec:3Dformal}

While the previous construction of a 3D surface code may seem arbitrary at first, we can naturally describe it in the language of chain complexes. This description allows us to further generalise the construction of codes for modular architectures and gives a strong intuition for which codes may or may not be constructed given the connectivity constraints.

First, consider a single loop of qubits with nearest neighbour connectivity between them. It is natural to view the loop as a classical repetition code where half of the qubits are assigned to be data qubits and the other half ancilla qubits as depicted in~\cref{RepetitionLoop}. As mentioned in~\cref{sec:chainC}, we may use a two-term chain complex $C = C_1 \xrightarrow[]{H_X} C_0$ to describe the repetition code, where the data and ancilla qubits are elements (chains) of $C_1$ and $C_0$ respectively and $H_X$ is the parity check matrix of the code, see~\cref{3DSurface}(a).

Moreover, the 2D layout of the loops can be considered as a two-term chain complex \mbox{$D = D_2 \xrightarrow[]{H_Z^T} D_1 \xrightarrow[]{H_X} D_0$} representing a 2D surface code. In this language, each loop is labeled as an element of either $D_2$, $D_1$ or $D_0$, which represent $Z$ checks (faces), data qubits (edges) or $X$ checks (vertices) respectively, see~\cref{3DSurface}(b).

\begin{figure}
    \centering
    \includegraphics[width=\linewidth]{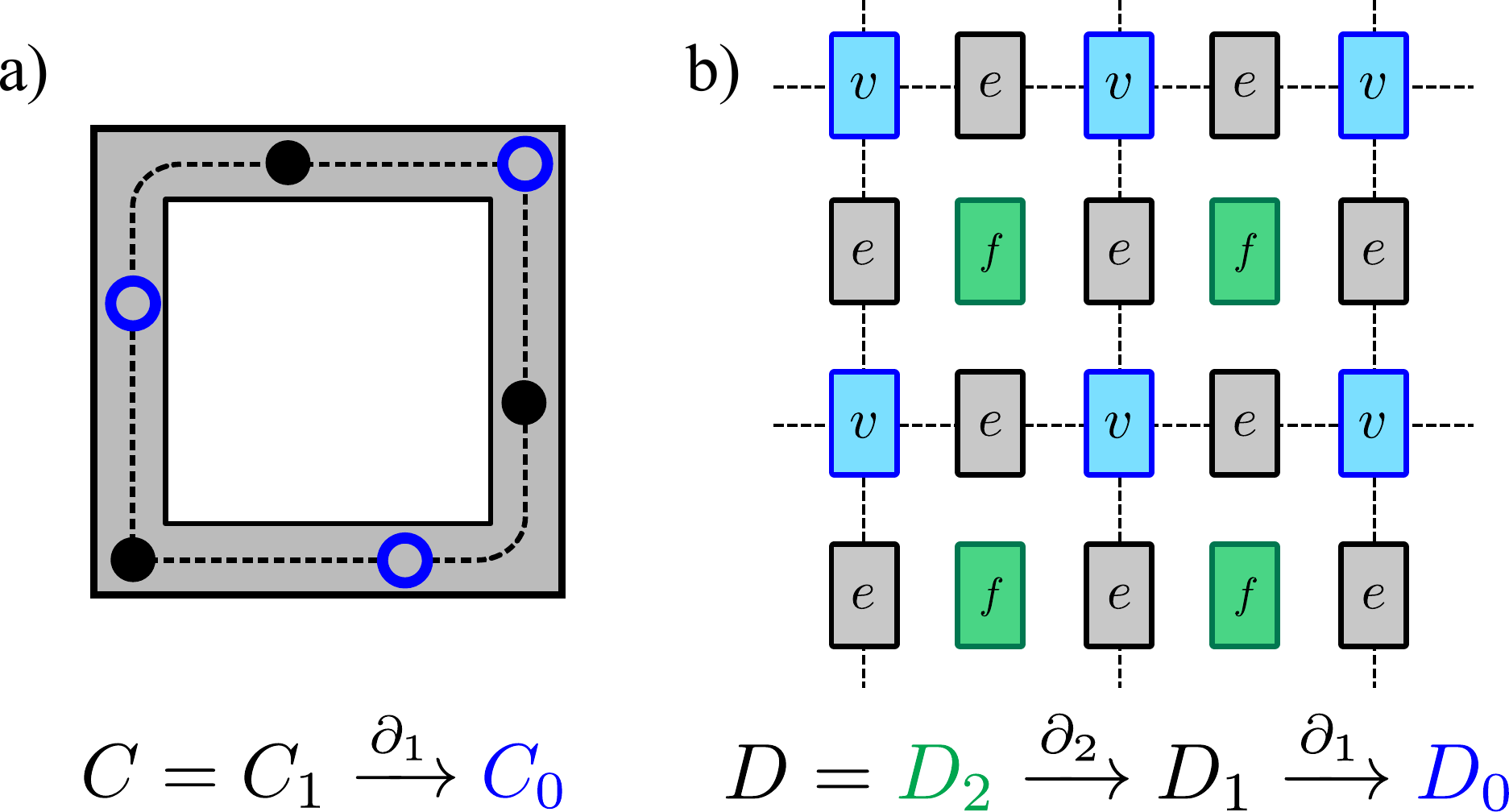}
    \caption{ A single loop corresponds to a two-term chain complex and a layout of loops to a three-term chain complex representing a repetition code (a) and a surface code (b) respectively.}
    \label{3DSurface}
\end{figure}

Using such a layout of loops we have effectively created a new code $\mathcal{E}$.
As outlined in~\cref{eq:TotalComplex}, $\calE{}$ is represented by a chain complex

\begin{equation}
    E{} = C{} \otimes D{},
\end{equation}
where $C$ corresponds to the code within the loop and $D$ corresponds to the code describing the layout of the loops.
\begin{figure}
    \centering
    \includegraphics[width=0.6\linewidth]{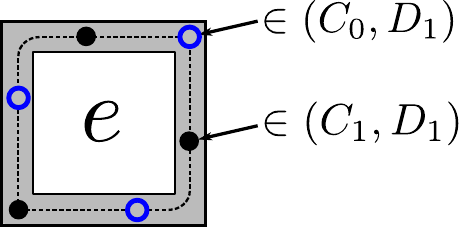}
    \caption{Our observations allow to readily identify to which vector spaces each qubit belongs to. The ancilla qubits of the repetition code are elements (chains) of $C_0$, while the data qubits are chains of $C_1$. The loop itself has an assignment as an edge loop, therefore, every qubit in it belongs to $D_1$.}
    \label{QubitAssignments}
\end{figure}
In more detail, we associate two vector spaces $C_{i}$ and $D_{j}$ to each qubit in our system. In this example, $C_{i}$ describes whether the qubit in the repetition code is an ancilla ($i=0$) or data ($i=1$) qubit, while $D_j$ describes whether the qubit belongs to the face ($j=2$), edge ($j=1$) or vertex ($j=0$) loop. See~\cref{QubitAssignments} for a schematic explanation of qubits in the edge loop.
Furthermore, we assign each qubit to a vector space $E_k$, with $k=i+j$, which form the sequence of vector spaces for a new four-term chain complex
\begin{equation}
    E = E_3 \xrightarrow[]{\partial_3} E_2 \xrightarrow[]{\partial_2} E_1 \xrightarrow[]{\partial_1} E_0,
\end{equation}
where 
\begin{align*}
    E_3 &= C_1 \otimes D_2, \\
    E_2 &= C_0 \otimes D_2 \oplus C_1 \otimes D_1, \\
    E_1 &= C_0 \otimes D_1 \oplus C_1 \otimes D_0, \\
    E_0 &= C_0 \otimes D_0.
\end{align*}
This chain complex $E$ describes a 3D surface code, as per the tensor product of a repetition code and a surface code \cite{zeng2019higher, higgott2022improved}. Here, for example, we can identify the data qubits with chains in $E_1$, and hence, $Z$ ($X$) stabilizers with chains in $E_2$ ($E_0$). Then, parity check matrices are given as boundary operators $\partial_2 = H_Z^T$ and $\partial_1 = H_X$.
Note that in this construction one of the boundaries of the surface is periodic and, also, that we are free to identify the data qubits with chains either in $E_2$ or $E_1$. This freedom corresponds to the choice of having the logical $X$ ($Z$) operators to be planar (string)-like on the 3D surface or the other way around.

As an example, consider an $L\times L$ surface code as the layout of the loops. It has parameters $\llbracket 2(L^2-L)+1, 1, L \rrbracket $. The respective chain complex $D$ has homology $H_1(D) \cong \mathbb{Z}_2$ and $H_0(D)$ is trivial. Similarly, consider a length $L$ classical repetition code inside the loop with parameters $[L, 1, L] $. The respective chain complex $\mathcal{C}$ has homology $H_1(C) \cong \mathbb{Z}_2$ and $H_0(C) \cong \mathbb{Z}_2$ since one of the checks is linearly dependent. 
Then for the chain complex of a 3D surface code $E = C \otimes D$, by identifying the elements of $E_1$ as data qubits, we find $n = \mathrm{dim }E_1 = \mathrm{dim } (C_1 \otimes D_0 \oplus C_0 \otimes D_1) = L\cdot (L^2-L) + L\cdot (2(L^2-L)+1)$, where $D_0$ and $C_0$ are vector spaces associated with $X$ parity checks.
The number of encoded qubits is given by the dimension of the first homology $H_1(E)$. We can compute it using the K\"unneth formula~\cref{eq:Kuenneth1},
\begin{multline*}
k= \mathrm{dim} H_1(E) = \\
= \mathrm{dim} (H_0(C) \otimes H_1(D) \oplus H_1(C) \otimes H_0(D)) = 1.
\end{multline*}
The distance of the code is still L, hence, the resulting 3D surface code has parameters $\llbracket 3L(L^2-L)+L, 1, L \rrbracket $. As an example, for $L=20$ the parameters are $\llbracket 22820, 1, 20 \rrbracket $.

Note that in this construction qubits that were previously data qubits may be reassigned to parity check qubits and vice versa. More importantly, the intra- and inter-modular connectivity requirements of $\calE$ correspond to the connectivity requirements given by codes $\mathcal{C}$ and $\calD{}$. We can prove this for general tensor products of chain complexes.

\begin{theorem}\label{prop:hgp-respects-modularity}
Let $C$ and $D$ be a two- or three-term chain complexes representing classical or quantum codes $\calC, \calD$ respectively. Let their boundaries $\partial^C$ and $\partial^D$ define the intra- and inter-modular connectivity respectively.
Then a quantum code $\calE$ corresponding to the chain complex $E = C \otimes D$ respects the connectivity constraints of the architecture.
\end{theorem}

\begin{proof}

The chain complex $E$ (corresponding to $\calE$) is at least a three-term chain complex given that both $C$ and $D$ are at least two-term chain complexes. Therefore, we can identify chains of $E_i$ with data qubits, $E_{i+1}$ with Z parity checks and $E_{i-1}$ with X parity checks. The required qubit connectivity of $\cal{E}$ is defined by its parity check matrices, which are given by the boundary operators 
\begin{equation}
\label{eq:proof1}
    \partial^E =  \partial^h \oplus \partial^v =  \mathrm{id}^C \otimes \partial^D \oplus \partial^C \otimes \mathrm{id}^D.
\end{equation}

Note that $i$-chains of $C$ label the qubit $c$ within each module and that $i$-chains of $D$ label each module $d$. We can ensure that $\mathcal{E}$ respects the connectivity constraints of the architecture if both terms of~\cref{eq:proof1} match the given qubit connectivity.
We do so by looking at both of the terms separately.

The basis elements are pairs of chains $(c,d) \in C\times D$ on which the first term acts as $\partial^h: (c, d) \rightarrow (\mathrm{id}^C(c),\partial^D(d))$ $\forall c,d$. By linear extension this defines a map on $C\otimes D$. 
It is equally stated that \textit{each} qubit $c$ in a module $d$ is connected to its \textit{respective} qubits $c$ in adjacent modules given by $\partial^D$ $\forall d$. This matches~\cref{def:modules} for the inter-modular connectivity. Similarly, the second term in~\cref{eq:proof1} defines a map that acts on basis elements as $\partial^v: (c, d) \rightarrow (\partial^C(c),\mathrm{id}^D(d))$ $\forall c,d$. This is equally stated that in \textit{each} module $d$ a qubit $c$ is connected to qubits $\partial^C(c)$ $\forall c$. This matches~\cref{def:qubits} for intra-modular connectivity. 
Therefore, the required qubit connectivity of $\cal{E}$ is given by some additive combination of terms which define the intra- and inter-modular connectivity respectively.  
\end{proof}

\noindent Furthermore, it is clear that the Theorem~\ref{prop:hgp-respects-modularity} still applies whenever boundaries $\partial^C$ and $\partial^D$ define any subgraphs of intra- and inter-modular connectivity respectively. If the subgraphs are proper then some connectivity that is allowed by the architecture is not required. Notice that our qubit assignment in the chain complex formalism completely matches the assignments given in~\cref{Tab1} if we identify the data qubits with the vector space $E_1$ and ignore qubits in $E_3$. Since the qubit assignments and the qubit interactions match between both perspectives, the stabilisers of the code match as well. This proves that our previous construction in~\cref{sec:3Dsurface} produced a 3D surface code.

The correspondence between modular architectures and quantum codes obtained from product constructions is very natural and gives an intuitive way of designing codes that obey architectural connectivity constraints. In the next sections we propose generalisations of this idea in a step by step fashion. First, we generalise the intra-modular connectivity by considering a less local code within each module. Then we similarly generalise the inter-modular connectivity. Finally, we elaborate on a more general product code constructions by allowing \emph{twists} between inter-modular connections.

\section{Hypergraph Product Codes}\label{sec:hgp}
In the previous section it was shown that the proposed formalization of the looped pipeline architecture enables to obtain a 3D surface code that can be viewed in a rigorous way as the tensor product of two chain complexes $C \otimes D{}$, where $C$ corresponds to the loop structure and $D$ to the (grid-like) layout. 

In this section, we further elaborate on the tensor product code construction and extend it to the setting of more general intra-modular and inter-modular connectivity.

\subsection{Generalised Intra-Modular Connectivity}

The first generalisation of the proposed formalization of the looped pipeline architecture is to replace the loops of qubits with modules admitting a more general intra-modular connectivity. Similarly to viewing the loops of qubits as repetition codes, we view this connectivity as a chain complex corresponding to some code $\calC{}$, for instance a simple classical linear block LDPC code. Note that in general this implies that a higher degree of intra-modular connectivity is needed.

Formally we consider a tensor product $E = C \otimes D$ of a chain complex $C$ corresponding to some classical or quantum code and a 3-term chain complex $D$ corresponding to a surface code, describing the overall layout of modules. This yields a 4 or 5-term chain complex $E$ which can be viewed as a surface code layout of modules, where each module is replaced by an arbitrary code given by the chain complex $C$. In general this does not yield a nice 3D geometry as in the more simple 3D surface code case.

As an explicit example, we consider intra-modular connectivity that corresponds to a classical linear code. We obtain a code by generating a random sparse parity check matrix with dimensions $51 \times 60$. Through exhaustive search over all codewords we find the code parameters $[60, 9, 20]$. Its maximum row or column weight is $8$, but on average each check has a support of $\approx 5$ bits. See~\cite{githurepo2} for the full parity check matrix. The respective chain complex $C$ has homology dimension $\mathrm{dim}(H_1(C))=9$ since it encodes $9$ bits and $H_0(C)$ is trivial as all parity checks are linearly independent. The homological properties of the chain complex $D$ associated with the surface code describing the inter-modular layout are given in the previous example. Note that $H_2(D)$ is trivial and for this example we choose a surface code with length $L=20$.
Then for the resulting chain complex $E = C \otimes D$, by identifying the elements of $E_2$ as data qubits, we find $n = \mathrm{dim}E_2 = \mathrm{dim} (C_1 \otimes D_1 \oplus C_0 \otimes D_2) = 65040$, where $D_2$ and $C_0$ are vector spaces associated with $Z$ and $X$ parity checks respectively.
The number of encoded qubits is given by the dimension of the second homology $H_2(E)$. We can compute it using the K\"unneth formula (\cref{eq:Kuenneth1}),
\begin{multline*}
k= \mathrm{dim} H_2(E) = \\
= \mathrm{dim} (H_1(C) \otimes H_1(D) \oplus H_0(C) \otimes H_2(D)) = 9.
\end{multline*}

Generally, finding the distance of the code is not trivial, thus usually Monte-Carlo simulations or similar approaches are employed. For general hypergraph products of arbitrary length chain complexes, Zeng and Pryadko~\cite{zeng2019higher} proposed methods to compute upper and lower bounds on the distance of the hypergraph product complex from the distances of the individual complexes in the product. 
Moreover, for the special case where one of the chain complexes in the hypergraph product is a 1-term complex, given by a binary check matrix, Zeng and Pryadko show that their result allows to compute the distance exactly. Recall that the homological distance $d_i$ is the minimum Hamming weight of a non-trivial representative in the $i$-th homology group.

\begin{theorem}{(\cite{zeng2019higher})}\label{theorem:hp-distance-special-case}
	Let $A$ be a $m$ chain complex with distances $d_i$ for $0 \leq i \leq m$ and let $B$ be a 2-term chain complex. Then, 
	$$d_i(A \otimes B) = \text{min}(d_{i-1}(A)d_1(B), d_i(A)d_0(B)).$$
\end{theorem}

Applying~\cref{theorem:hp-distance-special-case} to the previous example we find the $Z$ distance of the code is
\begin{equation*}
    d_2(E) = \text{min}(d_{1}(D)d_1(C), d_2(D)d_0(C)) = 400,
\end{equation*}
where by convention $d_i(A) = \infty$ if $H_i(A)$ is trivial. And we find the $X$ distance of the code using cohomology
\begin{equation*}
    d^2(E) = \text{min}(d^{1}(D)d^1(C), d^2(D)d^0(C)) = 20.
\end{equation*}
Therefore, the code $E = C\otimes D$ has parameters $\llbracket 65040, 9, 20 \rrbracket$. 
Given that this code is constructed as a tensor product of a quantum and a classical code, and moreover, that the $Z$ distance is $d^2= 400$, we choose to compare it to the traditional 3D surface code with distance $20$ as it matches both the $X$ and $Z$ distance parameters. In comparison to the 3D surface code (with full parameters given in Subsection~\ref{sec:3Dformal}) one can see that our example code encodes 9 times more qubits at the cost of increasing the overall number of physical qubits by a factor of about 3. We would expect such favourable trade-offs for architectures with higher qubit connectivity.

\begin{figure*}[t]
    \centering
    \includegraphics[width=0.5\linewidth]{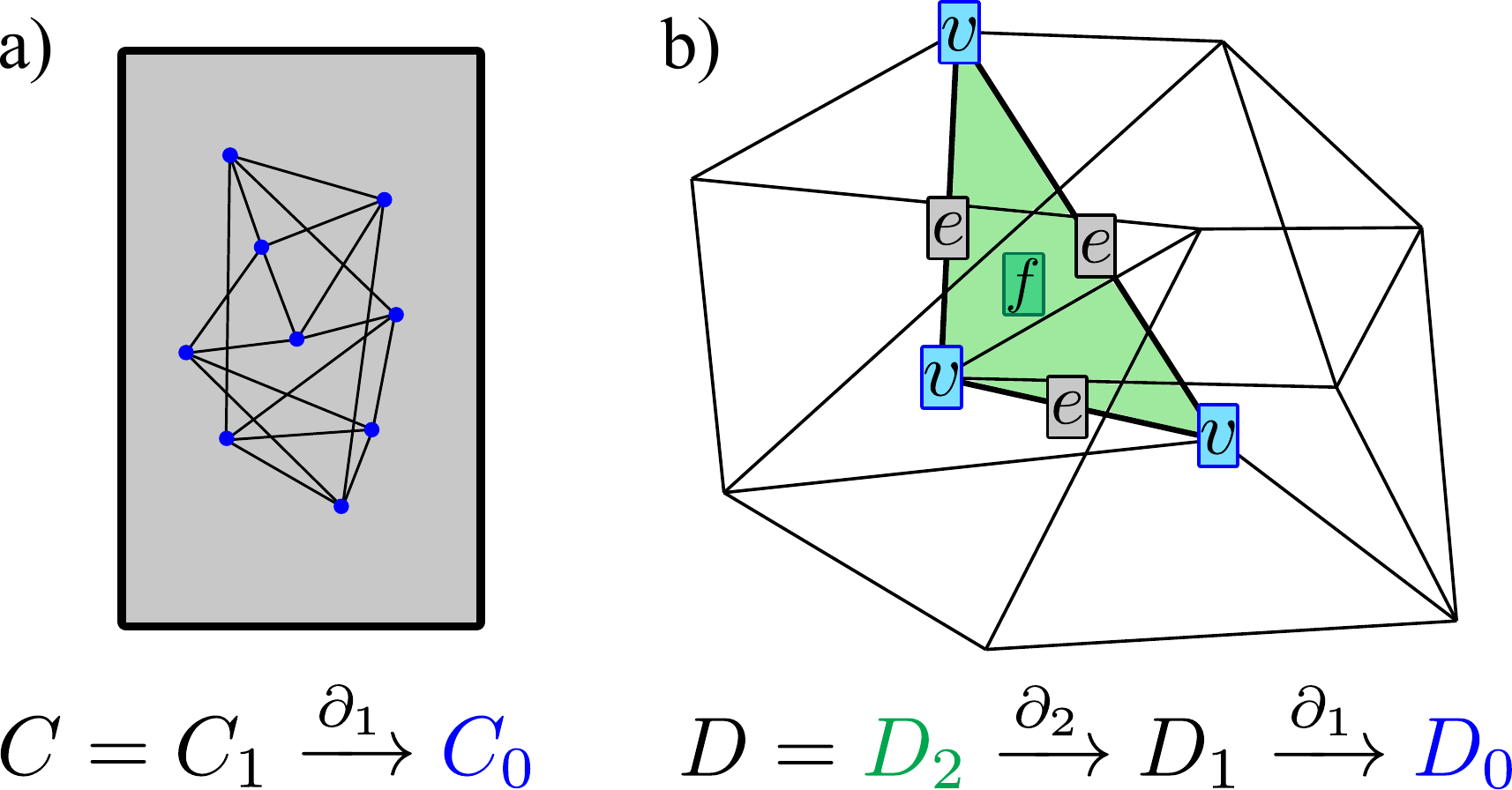}
    \caption{(a) We represent the intra-modular qubit connectivity with a chain complex $C$, where the qubits sit on all $i$-chains. (b) Similarly, we represent the inter-modular connectivity with a chain complex $D$, where modules are placed on every $i$-chain. Some elements of the chain complex are highlighted. The code $E = C \otimes D$ fully respects the connectivity constraints of the architecture (\autoref{prop:hgp-respects-modularity}).}
    \label{fig:genIdea}
\end{figure*}

This idea can be generalised to more connected modules, by considering a connectivity inside each module that corresponds to a quantum code, i.e., a three-term chain complex.
Note that this may imply an even higher degree of intra-modular connectivity and a higher number of qubits associated with each module. The latter implies that there are more connections between modules as per~\cref{def:modules}. 
Due to the generality of the chain complex formalism, this construction is equivalent to the previous case, with the only difference that the resulting product is a 5-term chain complex.

\subsection{Beyond Planar Surface Layouts}\label{sec:beyond}
In a similar fashion to the discussion on higher intra-modular connectivity, we can generalise the layout of the modules. This can be done by considering layouts of modules corresponding to a general code given by a chain complex $D$. We examine 3-term chain complexes describing quantum codes, but the same reasoning is applicable to classical codes. The 2D grid layout (corresponding to a surface code) has the advantage of planarity, which implies a sparse nearest-neighbour connectivity. However, the same planarity clearly limits the code parameters of the derived product codes. Additionally, some architectures (e.g. based on photonic links between modules) might not even be subject to nearest-neighbour communication constraints and hence would have no benefit from a planar layout of modules. In such cases as long as it holds that modules communicate with a few other modules each, it is not crucial that they are located close to each other physically and we can use less geometrically local codes to describe the layout of the modules.
We consider a modular layout given by a chain complex

\[D = D_F \xrightarrow[]{\bdry_2} D_E \xrightarrow[]{\bdry_1} D_V. \]

For illustrative purpose we call the vector spaces of the chain complex the spaces of \emph{faces}, \emph{edges}, and \emph{vertices}. $D$ can be made to correspond to the given inter-modular connectivity by associating modules to each $i$-chain of $D$. That is, the spaces $D_F,D_E,D_V$ are associated with a set of modules $\set{m_1, \dots, m_{\abs{V}+\abs{E}+\abs{F}}}$ and the modules are connected corresponding to the boundary maps of $D$. In other words, the differentials $\bdry_2, \bdry_1$ are the incidence matrices associating faces and edges, and edges and vertices respectively. 
If $C$ denotes the chain complex describing the intra-modular connectivity, then by~\cref{prop:hgp-respects-modularity} the tensor product complex 
$C \otimes D$ represents a quantum code that respects the connectivity of the overall architecture. This idea is illustrated in~\cref{fig:genIdea}.

We are now able to formally describe codes arising from modular architectures using the following recipe: we take an arbitrary CSS code representing the intra-modular connectivity (depending on the desired degree of connectivity, number of qubits, and code parameters) and another CSS code that represents the inter-modular connectivity. Then using the tensor product chain complex we construct a new code satisfying the architectural connectivity constraints. 
This can already give moderately good codes for a specific modular architecture, depending on the chosen seed codes and the allowed degree of connectivity between modules.

\subsection{From Codes to Modular Architectures}
\begin{figure}
    \centering
    \includegraphics[width=\linewidth]{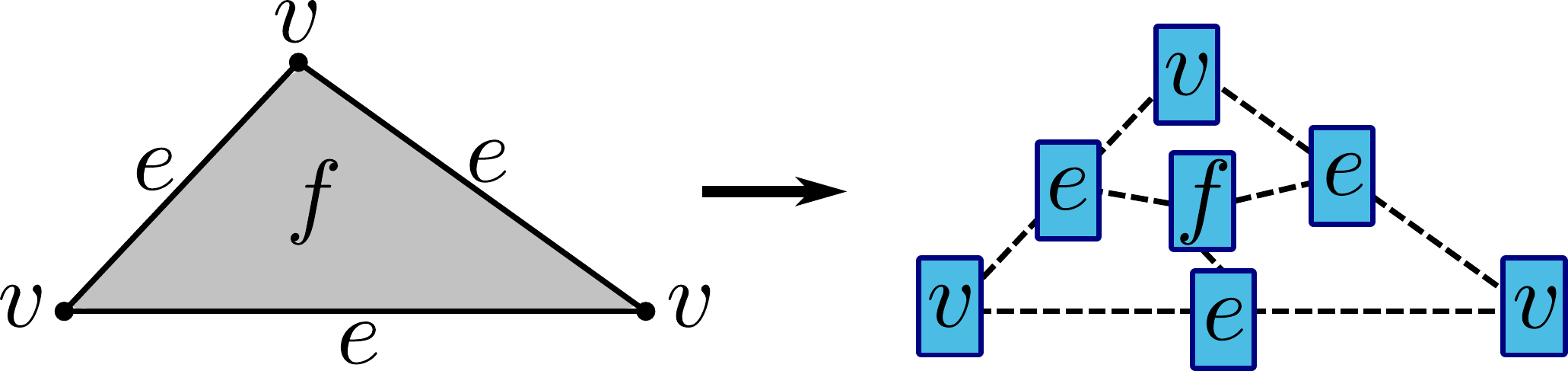}
    \caption{We can take any chain complex (on the left) to describe the inter-modular connectivity of our architecture by replacing every element of the chain complex by a module (on the right).}
    \label{GenModuleComplex}
\end{figure}

Let us briefly point a slightly different view on the construction presented above. An architectures might allow an ``all-to-all'' connectivity between modules, or a connectivity that is not constrained other than requiring that each module may only be connected to a constant number of other modules. Then, we can ask the question of given these ``weak'' constraints, how should we arrange the modules in order to generate a good code tailored for the architecture? This can be done by choosing a code defined by a chain complex $C$ 
and mapping it to the respective connectivity graph as introduced in \ref{sec:chain-complexes}.This graph in turn, defines the overall modular layout as depicted in~\cref{GenModuleComplex}.

\section{Balanced Product Constructions}\label{sec:bp}
In the previous section we showed how modular architectures can be viewed as tensor products of chain complexes (that describe codes). This perspective corresponds to having intra-modular connectivity--the layout of qubits within each module--and inter-modular connectivity--the layout of modules themselves--being determined by such codes. In order to further generalise the previous constructions, we consider modular architectures where the inter-modular connections are not constrained to be between the respective qubits only. Instead, these connections can be tweaked (\emph{twisted}) in some way that will be dictated by the product construction. %
Hence, we redefine the inter-modular connectivity as follows

\begin{definition}
\label{def:modules2}
A module $M_k$ is \textit{connected} to a module $M_j$ if the architecture allows us to directly implement two-qubit entangling operations between a qubit $q_i \in M_k$ and its \textit{respective} qubit $q_l \in M_j$ for all $i,l$.

\end{definition}
\noindent Note that many quantum computing platforms that consider linking modules together with photonic links or similar, already allow this degree of freedom. For this setting, we consider more general products than standard hypergraph products. Specifically, we show that the newly defined inter-modular connectivity allows us to construct codes that can be described in the language of \emph{balanced product codes}~\cite{DBLP:journals/tit/BreuckmannE21}. As before, these codes fully respect the architectural connectivity constraints. Before proceeding to the proof, let us shortly introduce the notion of balanced products of chain complexes.

\subsection{Balanced Product Chain Complexes}\label{sec:bp-chain-compl}
Breuckmann and Eberhardt (BE) introduced \emph{balanced product codes}, which are analogously constructed to the balanced (or mixed) product of topological spaces~\cite{DBLP:journals/tit/BreuckmannE21}. This and related constructions can be used to construct \emph{asymptotically good} quantum codes~\cite{DBLP:journals/tit/BreuckmannE21, panteleev2022asymptotically, leverrier2022quantum, lin2022good,dinur2022good}.

In order to define the balanced product of chain complexes we need to discuss the balanced product of vector spaces. Let $V,W$ be vector spaces with a linear right and left action respectively of a finite group $G$. The balanced product is defined as the quotient

\[V\otimes_G W = V\otimes W/\angles{vg\otimes w - v \otimes gw},\]
where  $v \in V, w\in W$, and $g\in G$.
If $V$ and $W$ have bases $X$ and $Y$, and $G$ maps basis vectors to basis vectors then the basis of $V \otimes_G W$ is given by $X \times_G Y = X\times Y/\sim$. The equivalence relation $\sim$ is defined as \mbox{$(x,y) \sim (xg, g^{-1}y)$} for all $x\in X, y\in Y$, and $g\in G$. We can now extend this notion to chain complexes. Let $C$ and $D$ be chain complexes where $C$ has a linear right action and $D$ has a linear left action of a group $G$. %
The \emph{balanced product double complex} $C \boxtimes_G D$ is defined via 

\[(C \boxtimes_G D)_{p,q} = C_p \otimes_G D_q, \]
with horizontal and vertical differentials defined analogously to the double complex~\cref{eqn:doubleComplex}, that act on the quotients $C_p \otimes_G D_q$ of vector spaces $C_p$ and $D_q$.
The \emph{balanced product complex} is the corresponding total complex: 

\[C \otimes_G D = \text{Tot}(C\boxtimes_G D).\]
We limit the discussion to cases where the vector spaces $C_i, D_i$ are based and the action of $G$ restricts to an action on these bases~\cite{DBLP:journals/tit/BreuckmannE21}.
If $G$ is a finite group of odd order, BE~\cite{DBLP:journals/tit/BreuckmannE21} gave a version of the K\"unneth formula that can be applied to the balanced product complex:
\begin{equation}\label{eq:kunneth-bp}
H_n(C \otimes_G D) \cong \bigoplus_{p+q=n} H_p(C) \otimes_G H_q(D).
\end{equation}
We want to emphasize that for certain cases the balanced product is equivalent to related constructions such as the \emph{lifted product}~\cite{panteleev2021quantum, panteleev2021degenerate} and the \emph{fiber bundle construction}~\cite{hastings2021fiber}. %

\subsection{Architecture Tailored Codes From Balanced Products}

To prove that the more general inter-modular connectivity including twists allows us to construct better quantum codes than those constructed as hypergraph products, we consider cases where we can cast the balanced product $C \otimes_G D$ as a fiber bundle complex $B \otimes_\varphi D$. In this complex, $B$ denotes the base, $D$ denotes the fiber and $\varphi$ the connection that describes the \emph{twists} of the fiber along the base. 
Using the homological language, the connection $\varphi$ represents an automorphism on the fiber $D$ that alters the horizontal differentials of the double complex $B \boxtimes D$. Similarly to other products, the fiber bundle complex can be used to describe a quantum error correcting code once we identify the data qubits and the parity checks correspondingly. Such codes are called \emph{fiber bundle codes}~\cite{hastings2021fiber}. 
In our modular architecture setting, it is crucial to correctly identify the base and fiber chain complex to ensure that connections are twisted between modules only. This idea is depicted in~\cref{fig:connection-twists}. 
Note that, twisting connections in the ``other direction'' (i.e., the connections are twisted between layers of respective qubits across all modules) generally results in a low connection count in each module. For many quantum platforms, the intra-modular connections are considered to be faster to implement and less noisy, and, therefore, preferential over the qubit connections between distinct modules~\cite{ramette2023fault, campbell2007, Stephenson2020}.

BE showed that when $C$ is a two-term complex and $H$ is abelian and acts freely on the bases of each $C_i$, then there exists a connection $\varphi$ s.t. $C \otimes_G D = B \otimes_\varphi D$, where $B_i = C_i/\angles{cg-c}$~\cite{DBLP:journals/tit/BreuckmannE21}. Therefore, a wide range of codes constructed from balanced products can be recast into the language of fiber bundle codes. Here we restrict ourselves to these cases and cast the following theorem in terms of fiber bundle codes.

\begin{figure}
    \centering
    \includegraphics[width=\linewidth]{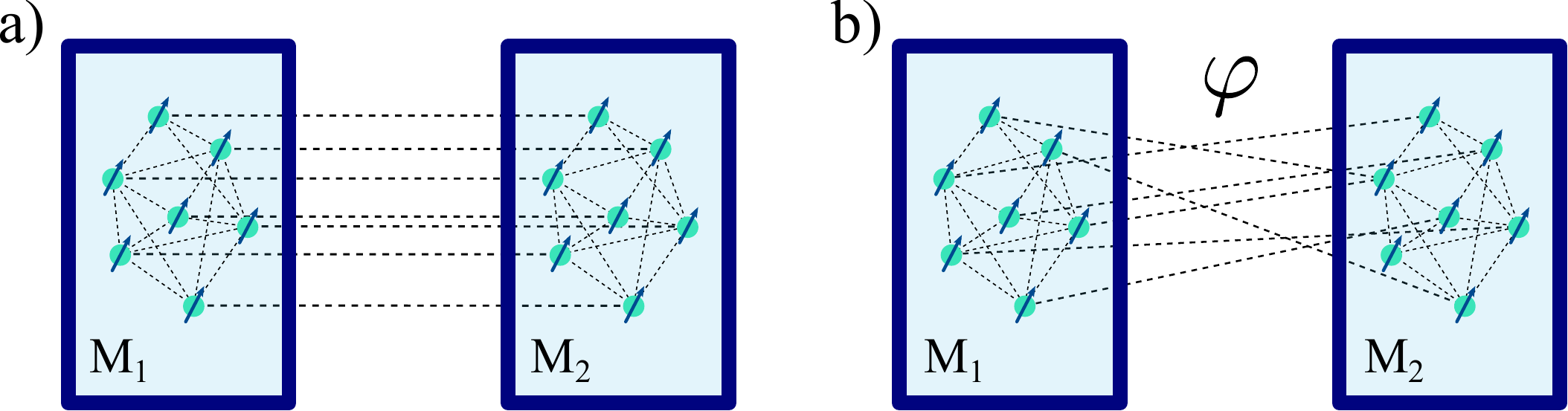}
    \caption{Allowing twists of inter-module connections allows to create better codes and yields a more general formulation.}
    \label{fig:connection-twists}
\end{figure}

\begin{theorem}\label{thm:bp-respects-architecture}
Let $D$ be a two-term and $C$ a two- or three-term chain complex. Let their boundaries $\partial^C$ and $\partial^D$ define the intra- and inter-modular connectivity as given in~\cref{def:qubits} and~\cref{def:modules2} respectively. Then, a fiber bundle code $\mathcal{E}$ corresponding to the chain complex $E = D \otimes_\varphi C$ respects the connectivity constraints of the architecture.
\end{theorem}

\begin{proof}
    The resulting chain complex $E$ (corresponding to $\calE$) is at least a three-term chain complex given that both $C$ and $D$ are at least two-term chain complexes. Therefore, we can identify chains of $E_i$ with data qubits, $E_{i+1}$ with Z parity checks and $E_{i-1}$ with X parity checks. The required qubit connectivity of $\cal{E}$ is defined by the parity check matrices $H^E$, which are given by the boundary operators 
\begin{equation}
\label{eq:proof2}
    \partial^E =  \partial^h \oplus \partial^v =  \partial_{\varphi} \oplus \mathrm{id}^D \otimes \partial^C,
\end{equation}
where $$\partial_{\varphi}(d_1 \otimes c) = \sum_{d_0 \in \partial^D d_1} d_0 \otimes \varphi(d_1, d_0)(c)$$
in which $d_i \in D_i$ and $c$ is any $i$-chain of $C$. Here, $\varphi$ denotes the connection and hence $\varphi(d_1, d_0)$ describes a specific element of the automorphism group acting on the fiber $C$.

Note that $i$-chains of $C$ label the qubit $c$ within each module and that $i$-chains of $D$ label each module $d$. We can ensure that $\mathcal{E}$ respects the connectivity constraints of the architecture if both terms of~\cref{eq:proof2} match the given qubit connectivity.
We do so, by looking at both of the terms separately.
The first term in~\cref{eq:proof2} describes that each qubit $c$ in a module $d_1$ is connected to qubits labeled $\varphi(d_1, d_0)c = c'$ in adjacent modules $d_0 \in \partial^D(d_1)$ $\forall d$. This connectivity requirement is fully satisfied by our new definition of inter-modular connectivity~\cref{def:modules2}. For this exact reason, we consider the base of the fiber bundle to represent the code describing the inter-modular connectivity as we want to \emph{twist} connections only between the modules.
The basis elements of the second term in~\cref{eq:proof2} are pairs of cells $(d,c) \in D \times C$ on which the boundary operator acts as $\partial^v: (d, c) \rightarrow (\mathrm{id}^D(d),\partial^C(c))$ $\forall d, c$. By linear extension this defines a map on $D \otimes C$. It is equally stated that in \textit{each} module $d$, a qubit $c$ is connected to qubits $\partial^C(c)$ $\forall c$. This matches our~\cref{def:qubits} for intra-modular connectivity. 
Therefore, the required qubit connectivity of $\cal{E}$ is given by some additive combination of terms which define the intra- and inter-modular connectivity respectively. 

\end{proof}

\cref{thm:bp-respects-architecture} thus establishes a correspondence between modular architectures (including inter-modular connectivity) and quantum codes that can be described using balanced product construction. While the proof considers balanced product codes that are equivalent to fiber bundle codes, we expect that a wider range of these codes respect the connectivity constraints---depending on the chosen group and group action.

As a simple example, we construct a balanced product code from two classical codes represented by chain complexes $C$ and $D$. The complex $C$ describes the intra-modular connectivity and corresponds to a $d=15$ cyclic repetition code as presented in~\cref{sec:looped-arch}. The code $D$ describes the inter-modular connectivity and was obtained by generating a random sparse parity-check matrix with dimensions $255 \times 450$ and a cyclic symmetry of order $15$. It encodes $\mathrm{dim}(H_1(D)) = 195$ bits.
Since both codes share the cyclic symmetry we take the balanced product over the group $\mathbb{Z}_{15}$. The product can be recast as a fiber bundle code and therefore satisfies~\cref{theorem:hp-distance-special-case}. The resulting code $D\otimes_{\mathbb{Z}_{15}}C$ has $n=705$ qubits and encodes at least $k= \mathrm{dim}(H_1(D/\mathbb{Z}_{15})) = 195/15 = 13$ logical qubits, where the calculation follows from the K\"unneth formula for fiber bundle codes~\cite{hastings2021fiber}. The $X$ and $Z$ parity checks have approximate average weight of $10$ and $6$, respectively. An exhaustive probabilistic distance search with \textit{QDistRnd}~\cite{QDistRnd0.8.5} showed that the code distance is at most $15$, we believe this bound is saturated. Hence, the balanced product code has expected parameters $\llbracket 705, 13, 15 \rrbracket$ and all qubits (including check qubits) can be partitioned into $47$ modules with $30$ qubits each. In comparison, encoding $13$ qubits into rotated surface codes with the same distance require $2925$ data qubits. The parity check matrix used for the code $D$ and parity check matrices $H_X$ and $H_Z$ for the balanced product code can be found at~\cite{githurepo2}.

\section{Conclusion}\label{sec:conclusion}
We have proposed a novel correspondence between two concepts from distinct rapidly evolving domains: QLDPC product code constructions and modular quantum computing architectures. Using tools from homological algebra that have been used in constructions of product codes, we give a novel way to view modular architectures as chain complexes and show that valid quantum codes that respect the given architectural constraints can be constructed from products of such chain complexes. 
Our results constitute an essential further step towards closing the gap between recent theoretical breakthroughs around asymptotically good quantum codes and practical applications of QLDPC codes.
Especially due to the formalisation of the looped pipeline architecture, we show practical relevance of our work.

As a direct extensions of this work it may be possible to generalise the considered constructions by allowing modules to have different inter-modular connectivity. This renders the construction more complex and product constructions as the ones used in this work might a priori not be applicable. Moreover, a further generalisation of our approach could be considered by investigating layouts of modular architectures. That is, by considering layouts of layouts of modules. It may be possible that such constructions can also be described using product constructions as described in this work but we leave an exact formulation open for future work. 
On a more practical note, it would be valuable to find small instances of QLDPC codes that can readily be realized on currently available architectures in order to draw comparisons to recent experimental breakthroughs around surface code realisations. 

In general, many open questions around practical aspects of QLDPC codes remain. An important example is how to do fault-tolerant logic on QLDPC codes. Investigating further areas of their potential and more practically relevant regimes around QLDPC codes is a crucial area of research towards scalable fault-tolerant quantum computing. 

\begin{acknowledgments}
The authors would like to thank Nikolas P. Breuckmann for insightful discussions throughout this project and comments on the first draft. Furthermore, the authors would like to thank Simon Benjamin, Zhenyu Cai, and Michael Fogarty for helpful discussions at the initial stages of this project. 
A.S. acknowledges support by the European Union’s Horizon 2020 research and innovation programme under grant agreement no. 951852 (QLSI).
L.B. acknowledges funding from the European Research Council (ERC) under the European Union’s Horizon 2020 research and innovation program (grant agreement No. 101001318) and this work was part of the Munich Quantum Valley, which is supported by the Bavarian state government with funds from the Hightech Agenda Bayern Plus, and has been supported by the BMWK on the basis of a decision by the German Bundestag through project QuaST, as well as by the BMK, BMDW, and the State of Upper Austria in the frame of the COMET program (managed by the FFG).
This research was funded in part by the UKRI grant number EP/R513295/1. For the purpose of Open Access, the author has applied a CC BY public copyright licence to any Author Accepted Manuscript (AAM) version arising from this submission.
\end{acknowledgments}
\bibliography{main}

\begin{thebibliography}{56}%
\makeatletter
\providecommand \@ifxundefined [1]{%
 \@ifx{#1\undefined}
}%
\providecommand \@ifnum [1]{%
 \ifnum #1\expandafter \@firstoftwo
 \else \expandafter \@secondoftwo
 \fi
}%
\providecommand \@ifx [1]{%
 \ifx #1\expandafter \@firstoftwo
 \else \expandafter \@secondoftwo
 \fi
}%
\providecommand \natexlab [1]{#1}%
\providecommand \enquote  [1]{``#1''}%
\providecommand \bibnamefont  [1]{#1}%
\providecommand \bibfnamefont [1]{#1}%
\providecommand \citenamefont [1]{#1}%
\providecommand \href@noop [0]{\@secondoftwo}%
\providecommand \href [0]{\begingroup \@sanitize@url \@href}%
\providecommand \@href[1]{\@@startlink{#1}\@@href}%
\providecommand \@@href[1]{\endgroup#1\@@endlink}%
\providecommand \@sanitize@url [0]{\catcode `\\12\catcode `\$12\catcode
  `\&12\catcode `\#12\catcode `\^12\catcode `\_12\catcode `\%12\relax}%
\providecommand \@@startlink[1]{}%
\providecommand \@@endlink[0]{}%
\providecommand \url  [0]{\begingroup\@sanitize@url \@url }%
\providecommand \@url [1]{\endgroup\@href {#1}{\urlprefix }}%
\providecommand \urlprefix  [0]{URL }%
\providecommand \Eprint [0]{\href }%
\providecommand \doibase [0]{https://doi.org/}%
\providecommand \selectlanguage [0]{\@gobble}%
\providecommand \bibinfo  [0]{\@secondoftwo}%
\providecommand \bibfield  [0]{\@secondoftwo}%
\providecommand \translation [1]{[#1]}%
\providecommand \BibitemOpen [0]{}%
\providecommand \bibitemStop [0]{}%
\providecommand \bibitemNoStop [0]{.\EOS\space}%
\providecommand \EOS [0]{\spacefactor3000\relax}%
\providecommand \BibitemShut  [1]{\csname bibitem#1\endcsname}%
\let\auto@bib@innerbib\@empty
\bibitem [{\citenamefont {Buyya}(1999)}]{buyya1999high}%
  \BibitemOpen
  \bibfield  {author} {\bibinfo {author} {\bibfnamefont {R.}~\bibnamefont
  {Buyya}},\ }\href@noop {} {\emph {\bibinfo {title} {High performance cluster
  computing: Programming and Applications, Volume 2}}}\ (\bibinfo  {publisher}
  {Prentice Hall},\ \bibinfo {year} {1999})\BibitemShut {NoStop}%
\bibitem [{\citenamefont {Singh}\ and\ \citenamefont
  {Reddy}(2015)}]{singh2015survey}%
  \BibitemOpen
  \bibfield  {author} {\bibinfo {author} {\bibfnamefont {D.}~\bibnamefont
  {Singh}}\ and\ \bibinfo {author} {\bibfnamefont {C.~K.}\ \bibnamefont
  {Reddy}},\ }\bibfield  {title} {\bibinfo {title} {A survey on platforms for
  big data analytics},\ }\href {https://doi.org/10.1186/s40537-014-0008-6}
  {\bibfield  {journal} {\bibinfo  {journal} {Journal of big data}\ }\textbf
  {\bibinfo {volume} {2}},\ \bibinfo {pages} {1} (\bibinfo {year}
  {2015})}\BibitemShut {NoStop}%
\bibitem [{\citenamefont {Monroe}\ \emph {et~al.}(2014)\citenamefont {Monroe},
  \citenamefont {Raussendorf}, \citenamefont {Ruthven}, \citenamefont {Brown},
  \citenamefont {Maunz}, \citenamefont {Duan},\ and\ \citenamefont
  {Kim}}]{Monroe2014}%
  \BibitemOpen
  \bibfield  {author} {\bibinfo {author} {\bibfnamefont {C.}~\bibnamefont
  {Monroe}}, \bibinfo {author} {\bibfnamefont {R.}~\bibnamefont {Raussendorf}},
  \bibinfo {author} {\bibfnamefont {A.}~\bibnamefont {Ruthven}}, \bibinfo
  {author} {\bibfnamefont {K.~R.}\ \bibnamefont {Brown}}, \bibinfo {author}
  {\bibfnamefont {P.}~\bibnamefont {Maunz}}, \bibinfo {author} {\bibfnamefont
  {L.-M.}\ \bibnamefont {Duan}},\ and\ \bibinfo {author} {\bibfnamefont
  {J.}~\bibnamefont {Kim}},\ }\bibfield  {title} {\bibinfo {title} {Large-scale
  modular quantum-computer architecture with atomic memory and photonic
  interconnects},\ }\href {https://doi.org/10.1103/PhysRevA.89.022317}
  {\bibfield  {journal} {\bibinfo  {journal} {Phys. Rev. A}\ }\textbf {\bibinfo
  {volume} {89}},\ \bibinfo {pages} {022317} (\bibinfo {year}
  {2014})}\BibitemShut {NoStop}%
\bibitem [{\citenamefont {Bartolucci}\ \emph {et~al.}(2021)\citenamefont
  {Bartolucci}, \citenamefont {Birchall}, \citenamefont {Bombin}, \citenamefont
  {Cable}, \citenamefont {Dawson}, \citenamefont {Gimeno-Segovia},
  \citenamefont {Johnston}, \citenamefont {Kieling}, \citenamefont {Nickerson},
  \citenamefont {Pant} \emph {et~al.}}]{Bartolucci2021}%
  \BibitemOpen
  \bibfield  {author} {\bibinfo {author} {\bibfnamefont {S.}~\bibnamefont
  {Bartolucci}}, \bibinfo {author} {\bibfnamefont {P.}~\bibnamefont
  {Birchall}}, \bibinfo {author} {\bibfnamefont {H.}~\bibnamefont {Bombin}},
  \bibinfo {author} {\bibfnamefont {H.}~\bibnamefont {Cable}}, \bibinfo
  {author} {\bibfnamefont {C.}~\bibnamefont {Dawson}}, \bibinfo {author}
  {\bibfnamefont {M.}~\bibnamefont {Gimeno-Segovia}}, \bibinfo {author}
  {\bibfnamefont {E.}~\bibnamefont {Johnston}}, \bibinfo {author}
  {\bibfnamefont {K.}~\bibnamefont {Kieling}}, \bibinfo {author} {\bibfnamefont
  {N.}~\bibnamefont {Nickerson}}, \bibinfo {author} {\bibfnamefont
  {M.}~\bibnamefont {Pant}}, \emph {et~al.},\ }\bibfield  {title} {\bibinfo
  {title} {Fusion-based quantum computation},\ }\href
  {https://arxiv.org/abs/2101.09310} {\bibfield  {journal} {\bibinfo  {journal}
  {arXiv preprint arXiv:2101.09310}\ } (\bibinfo {year} {2021})}\BibitemShut
  {NoStop}%
\bibitem [{\citenamefont {Gold}\ \emph {et~al.}(2021)\citenamefont {Gold},
  \citenamefont {Paquette}, \citenamefont {Stockklauser}, \citenamefont
  {Reagor}, \citenamefont {Alam}, \citenamefont {Bestwick}, \citenamefont
  {Didier}, \citenamefont {Nersisyan}, \citenamefont {Oruc}, \citenamefont
  {Razavi} \emph {et~al.}}]{gold2021}%
  \BibitemOpen
  \bibfield  {author} {\bibinfo {author} {\bibfnamefont {A.}~\bibnamefont
  {Gold}}, \bibinfo {author} {\bibfnamefont {J.}~\bibnamefont {Paquette}},
  \bibinfo {author} {\bibfnamefont {A.}~\bibnamefont {Stockklauser}}, \bibinfo
  {author} {\bibfnamefont {M.~J.}\ \bibnamefont {Reagor}}, \bibinfo {author}
  {\bibfnamefont {M.~S.}\ \bibnamefont {Alam}}, \bibinfo {author}
  {\bibfnamefont {A.}~\bibnamefont {Bestwick}}, \bibinfo {author}
  {\bibfnamefont {N.}~\bibnamefont {Didier}}, \bibinfo {author} {\bibfnamefont
  {A.}~\bibnamefont {Nersisyan}}, \bibinfo {author} {\bibfnamefont
  {F.}~\bibnamefont {Oruc}}, \bibinfo {author} {\bibfnamefont {A.}~\bibnamefont
  {Razavi}}, \emph {et~al.},\ }\bibfield  {title} {\bibinfo {title}
  {Entanglement across separate silicon dies in a modular superconducting qubit
  device},\ }\href {https://doi.org/10.1038/s41534-021-00484-1} {\bibfield
  {journal} {\bibinfo  {journal} {npj Quantum Information}\ }\textbf {\bibinfo
  {volume} {7}},\ \bibinfo {pages} {1} (\bibinfo {year} {2021})}\BibitemShut
  {NoStop}%
\bibitem [{\citenamefont {Hucul}\ \emph {et~al.}(2015)\citenamefont {Hucul},
  \citenamefont {Inlek}, \citenamefont {Vittorini}, \citenamefont {Crocker},
  \citenamefont {Debnath}, \citenamefont {Clark},\ and\ \citenamefont
  {Monroe}}]{hucul2015}%
  \BibitemOpen
  \bibfield  {author} {\bibinfo {author} {\bibfnamefont {D.}~\bibnamefont
  {Hucul}}, \bibinfo {author} {\bibfnamefont {I.~V.}\ \bibnamefont {Inlek}},
  \bibinfo {author} {\bibfnamefont {G.}~\bibnamefont {Vittorini}}, \bibinfo
  {author} {\bibfnamefont {C.}~\bibnamefont {Crocker}}, \bibinfo {author}
  {\bibfnamefont {S.}~\bibnamefont {Debnath}}, \bibinfo {author} {\bibfnamefont
  {S.~M.}\ \bibnamefont {Clark}},\ and\ \bibinfo {author} {\bibfnamefont
  {C.}~\bibnamefont {Monroe}},\ }\bibfield  {title} {\bibinfo {title} {Modular
  entanglement of atomic qubits using photons and phonons},\ }\href
  {https://doi.org/10.1038/nphys3150} {\bibfield  {journal} {\bibinfo
  {journal} {Nature Physics}\ }\textbf {\bibinfo {volume} {11}},\ \bibinfo
  {pages} {37} (\bibinfo {year} {2015})}\BibitemShut {NoStop}%
\bibitem [{\citenamefont {Buonacorsi}\ \emph {et~al.}(2019)\citenamefont
  {Buonacorsi}, \citenamefont {Cai}, \citenamefont {Ramirez}, \citenamefont
  {Willick}, \citenamefont {Walker}, \citenamefont {Li}, \citenamefont {Shaw},
  \citenamefont {Xu}, \citenamefont {Benjamin},\ and\ \citenamefont
  {Baugh}}]{Buonacorsi_2019}%
  \BibitemOpen
  \bibfield  {author} {\bibinfo {author} {\bibfnamefont {B.}~\bibnamefont
  {Buonacorsi}}, \bibinfo {author} {\bibfnamefont {Z.}~\bibnamefont {Cai}},
  \bibinfo {author} {\bibfnamefont {E.~B.}\ \bibnamefont {Ramirez}}, \bibinfo
  {author} {\bibfnamefont {K.~S.}\ \bibnamefont {Willick}}, \bibinfo {author}
  {\bibfnamefont {S.~M.}\ \bibnamefont {Walker}}, \bibinfo {author}
  {\bibfnamefont {J.}~\bibnamefont {Li}}, \bibinfo {author} {\bibfnamefont
  {B.~D.}\ \bibnamefont {Shaw}}, \bibinfo {author} {\bibfnamefont
  {X.}~\bibnamefont {Xu}}, \bibinfo {author} {\bibfnamefont {S.~C.}\
  \bibnamefont {Benjamin}},\ and\ \bibinfo {author} {\bibfnamefont
  {J.}~\bibnamefont {Baugh}},\ }\bibfield  {title} {\bibinfo {title} {Network
  architecture for a topological quantum computer in silicon},\ }\href
  {https://doi.org/10.1088/2058-9565/aaf3c4} {\bibfield  {journal} {\bibinfo
  {journal} {Quantum Science and Technology}\ }\textbf {\bibinfo {volume}
  {4}},\ \bibinfo {pages} {025003} (\bibinfo {year} {2019})}\BibitemShut
  {NoStop}%
\bibitem [{\citenamefont {Bravyi}\ \emph {et~al.}(2022)\citenamefont {Bravyi},
  \citenamefont {Dial}, \citenamefont {Gambetta}, \citenamefont {Gil},\ and\
  \citenamefont {Nazario}}]{bravyi2022future}%
  \BibitemOpen
  \bibfield  {author} {\bibinfo {author} {\bibfnamefont {S.}~\bibnamefont
  {Bravyi}}, \bibinfo {author} {\bibfnamefont {O.}~\bibnamefont {Dial}},
  \bibinfo {author} {\bibfnamefont {J.~M.}\ \bibnamefont {Gambetta}}, \bibinfo
  {author} {\bibfnamefont {D.}~\bibnamefont {Gil}},\ and\ \bibinfo {author}
  {\bibfnamefont {Z.}~\bibnamefont {Nazario}},\ }\bibfield  {title} {\bibinfo
  {title} {The future of quantum computing with superconducting qubits},\
  }\href {https://arxiv.org/abs/2209.06841} {\bibfield  {journal} {\bibinfo
  {journal} {arXiv preprint arXiv:2209.06841}\ } (\bibinfo {year}
  {2022})}\BibitemShut {NoStop}%
\bibitem [{\citenamefont {{Brown}}(2022)}]{brown2022modular}%
  \BibitemOpen
  \bibfield  {author} {\bibinfo {author} {\bibfnamefont {K.}~\bibnamefont
  {{Brown}}},\ }\bibfield  {title} {\bibinfo {title} {{Modular architectures
  for quantum computers}},\ }in\ \href
  {https://ui.adsabs.harvard.edu/abs/2022APS..MARF37006B} {\emph {\bibinfo
  {booktitle} {APS March Meeting Abstracts}}},\ \bibinfo {series} {APS Meeting
  Abstracts}, Vol.\ \bibinfo {volume} {2022}\ (\bibinfo {year} {2022})\ p.\
  \bibinfo {pages} {F37.006}\BibitemShut {NoStop}%
\bibitem [{\citenamefont {Bombin}\ \emph {et~al.}(2021)\citenamefont {Bombin},
  \citenamefont {Kim}, \citenamefont {Litinski}, \citenamefont {Nickerson},
  \citenamefont {Pant}, \citenamefont {Pastawski}, \citenamefont {Roberts},\
  and\ \citenamefont {Rudolph}}]{bombin2021interleaving}%
  \BibitemOpen
  \bibfield  {author} {\bibinfo {author} {\bibfnamefont {H.}~\bibnamefont
  {Bombin}}, \bibinfo {author} {\bibfnamefont {I.~H.}\ \bibnamefont {Kim}},
  \bibinfo {author} {\bibfnamefont {D.}~\bibnamefont {Litinski}}, \bibinfo
  {author} {\bibfnamefont {N.}~\bibnamefont {Nickerson}}, \bibinfo {author}
  {\bibfnamefont {M.}~\bibnamefont {Pant}}, \bibinfo {author} {\bibfnamefont
  {F.}~\bibnamefont {Pastawski}}, \bibinfo {author} {\bibfnamefont
  {S.}~\bibnamefont {Roberts}},\ and\ \bibinfo {author} {\bibfnamefont
  {T.}~\bibnamefont {Rudolph}},\ }\bibfield  {title} {\bibinfo {title}
  {Interleaving: Modular architectures for fault-tolerant photonic quantum
  computing},\ }\href {https://arxiv.org/abs/2103.08612} {\bibfield  {journal}
  {\bibinfo  {journal} {arXiv preprint arXiv:2103.08612}\ } (\bibinfo {year}
  {2021})}\BibitemShut {NoStop}%
\bibitem [{\citenamefont {Ramette}\ \emph {et~al.}(2023)\citenamefont
  {Ramette}, \citenamefont {Sinclair}, \citenamefont {Breuckmann},\ and\
  \citenamefont {Vuleti{\'c}}}]{ramette2023fault}%
  \BibitemOpen
  \bibfield  {author} {\bibinfo {author} {\bibfnamefont {J.}~\bibnamefont
  {Ramette}}, \bibinfo {author} {\bibfnamefont {J.}~\bibnamefont {Sinclair}},
  \bibinfo {author} {\bibfnamefont {N.~P.}\ \bibnamefont {Breuckmann}},\ and\
  \bibinfo {author} {\bibfnamefont {V.}~\bibnamefont {Vuleti{\'c}}},\
  }\bibfield  {title} {\bibinfo {title} {Fault-tolerant connection of
  error-corrected qubits with noisy links},\ }\href
  {https://arxiv.org/abs/2302.01296} {\bibfield  {journal} {\bibinfo  {journal}
  {arXiv preprint arXiv:2302.01296}\ } (\bibinfo {year} {2023})}\BibitemShut
  {NoStop}%
\bibitem [{\citenamefont {Niu}\ \emph {et~al.}(2023)\citenamefont {Niu},
  \citenamefont {Zhang}, \citenamefont {Liu}, \citenamefont {Qiu},
  \citenamefont {Huang}, \citenamefont {Huang}, \citenamefont {Jia},
  \citenamefont {Liu}, \citenamefont {Tao}, \citenamefont {Wei}, \citenamefont
  {Zhou}, \citenamefont {Zou}, \citenamefont {Chen}, \citenamefont {Deng},
  \citenamefont {Deng}, \citenamefont {Hu}, \citenamefont {Hu}, \citenamefont
  {Li}, \citenamefont {Tan}, \citenamefont {Xu}, \citenamefont {Yan},
  \citenamefont {Yan}, \citenamefont {Liu}, \citenamefont {Zhong},
  \citenamefont {Cleland},\ and\ \citenamefont {Yu}}]{niu2023low}%
  \BibitemOpen
  \bibfield  {author} {\bibinfo {author} {\bibfnamefont {J.}~\bibnamefont
  {Niu}}, \bibinfo {author} {\bibfnamefont {L.}~\bibnamefont {Zhang}}, \bibinfo
  {author} {\bibfnamefont {Y.}~\bibnamefont {Liu}}, \bibinfo {author}
  {\bibfnamefont {J.}~\bibnamefont {Qiu}}, \bibinfo {author} {\bibfnamefont
  {W.}~\bibnamefont {Huang}}, \bibinfo {author} {\bibfnamefont
  {J.}~\bibnamefont {Huang}}, \bibinfo {author} {\bibfnamefont
  {H.}~\bibnamefont {Jia}}, \bibinfo {author} {\bibfnamefont {J.}~\bibnamefont
  {Liu}}, \bibinfo {author} {\bibfnamefont {Z.}~\bibnamefont {Tao}}, \bibinfo
  {author} {\bibfnamefont {W.}~\bibnamefont {Wei}}, \bibinfo {author}
  {\bibfnamefont {Y.}~\bibnamefont {Zhou}}, \bibinfo {author} {\bibfnamefont
  {W.}~\bibnamefont {Zou}}, \bibinfo {author} {\bibfnamefont {Y.}~\bibnamefont
  {Chen}}, \bibinfo {author} {\bibfnamefont {X.}~\bibnamefont {Deng}}, \bibinfo
  {author} {\bibfnamefont {X.}~\bibnamefont {Deng}}, \bibinfo {author}
  {\bibfnamefont {C.}~\bibnamefont {Hu}}, \bibinfo {author} {\bibfnamefont
  {L.}~\bibnamefont {Hu}}, \bibinfo {author} {\bibfnamefont {J.}~\bibnamefont
  {Li}}, \bibinfo {author} {\bibfnamefont {D.}~\bibnamefont {Tan}}, \bibinfo
  {author} {\bibfnamefont {Y.}~\bibnamefont {Xu}}, \bibinfo {author}
  {\bibfnamefont {F.}~\bibnamefont {Yan}}, \bibinfo {author} {\bibfnamefont
  {T.}~\bibnamefont {Yan}}, \bibinfo {author} {\bibfnamefont {S.}~\bibnamefont
  {Liu}}, \bibinfo {author} {\bibfnamefont {Y.}~\bibnamefont {Zhong}}, \bibinfo
  {author} {\bibfnamefont {A.~N.}\ \bibnamefont {Cleland}},\ and\ \bibinfo
  {author} {\bibfnamefont {D.}~\bibnamefont {Yu}},\ }\bibfield  {title}
  {\bibinfo {title} {Low-loss interconnects for modular superconducting quantum
  processors},\ }\href {https://doi.org/10.1038/s41928-023-00925-z} {\bibfield
  {journal} {\bibinfo  {journal} {Nature Electronics}\ }\textbf {\bibinfo
  {volume} {6}},\ \bibinfo {pages} {235} (\bibinfo {year} {2023})}\BibitemShut
  {NoStop}%
\bibitem [{\citenamefont {Bartolucci}\ \emph {et~al.}(2023)\citenamefont
  {Bartolucci}, \citenamefont {Birchall}, \citenamefont {Bombin}, \citenamefont
  {Cable}, \citenamefont {Dawson}, \citenamefont {Gimeno-Segovia},
  \citenamefont {Johnston}, \citenamefont {Kieling}, \citenamefont {Nickerson},
  \citenamefont {Pant} \emph {et~al.}}]{bartolucci2023fusion}%
  \BibitemOpen
  \bibfield  {author} {\bibinfo {author} {\bibfnamefont {S.}~\bibnamefont
  {Bartolucci}}, \bibinfo {author} {\bibfnamefont {P.}~\bibnamefont
  {Birchall}}, \bibinfo {author} {\bibfnamefont {H.}~\bibnamefont {Bombin}},
  \bibinfo {author} {\bibfnamefont {H.}~\bibnamefont {Cable}}, \bibinfo
  {author} {\bibfnamefont {C.}~\bibnamefont {Dawson}}, \bibinfo {author}
  {\bibfnamefont {M.}~\bibnamefont {Gimeno-Segovia}}, \bibinfo {author}
  {\bibfnamefont {E.}~\bibnamefont {Johnston}}, \bibinfo {author}
  {\bibfnamefont {K.}~\bibnamefont {Kieling}}, \bibinfo {author} {\bibfnamefont
  {N.}~\bibnamefont {Nickerson}}, \bibinfo {author} {\bibfnamefont
  {M.}~\bibnamefont {Pant}}, \emph {et~al.},\ }\bibfield  {title} {\bibinfo
  {title} {Fusion-based quantum computation},\ }\href
  {https://www.nature.com/articles/s41467-023-36493-1} {\bibfield  {journal}
  {\bibinfo  {journal} {Nature Communications}\ }\textbf {\bibinfo {volume}
  {14}},\ \bibinfo {pages} {912} (\bibinfo {year} {2023})}\BibitemShut
  {NoStop}%
\bibitem [{\citenamefont {Preskill}(1998)}]{Preskill1997}%
  \BibitemOpen
  \bibfield  {author} {\bibinfo {author} {\bibfnamefont {J.}~\bibnamefont
  {Preskill}},\ }\bibfield  {title} {\bibinfo {title} {Fault-tolerant quantum
  computation},\ }in\ \href@noop {} {\emph {\bibinfo {booktitle} {Introduction
  to quantum computation and information}}}\ (\bibinfo  {publisher} {World
  Scientific},\ \bibinfo {year} {1998})\ pp.\ \bibinfo {pages}
  {213--269}\BibitemShut {NoStop}%
\bibitem [{\citenamefont {Lidar}\ and\ \citenamefont
  {Brun}(2013)}]{lidar_quantum_2013}%
  \BibitemOpen
  \bibfield  {author} {\bibinfo {author} {\bibfnamefont {D.~A.}\ \bibnamefont
  {Lidar}}\ and\ \bibinfo {author} {\bibfnamefont {T.~A.}\ \bibnamefont
  {Brun}},\ }\href@noop {} {\emph {\bibinfo {title} {Quantum error
  correction}}}\ (\bibinfo  {publisher} {Cambridge university press},\ \bibinfo
  {year} {2013})\BibitemShut {NoStop}%
\bibitem [{\citenamefont {Terhal}(2015)}]{Terhal2015}%
  \BibitemOpen
  \bibfield  {author} {\bibinfo {author} {\bibfnamefont {B.~M.}\ \bibnamefont
  {Terhal}},\ }\bibfield  {title} {\bibinfo {title} {Quantum error correction
  for quantum memories},\ }\href {https://doi.org/10.1103/RevModPhys.87.307}
  {\bibfield  {journal} {\bibinfo  {journal} {Rev. Mod. Phys.}\ }\textbf
  {\bibinfo {volume} {87}},\ \bibinfo {pages} {307} (\bibinfo {year}
  {2015})}\BibitemShut {NoStop}%
\bibitem [{\citenamefont {Panteleev}\ and\ \citenamefont
  {Kalachev}(2022)}]{panteleev2022asymptotically}%
  \BibitemOpen
  \bibfield  {author} {\bibinfo {author} {\bibfnamefont {P.}~\bibnamefont
  {Panteleev}}\ and\ \bibinfo {author} {\bibfnamefont {G.}~\bibnamefont
  {Kalachev}},\ }\bibfield  {title} {\bibinfo {title} {Asymptotically good
  quantum and locally testable classical ldpc codes},\ }in\ \href
  {https://doi.org/10.1145/3519935.3520017} {\emph {\bibinfo {booktitle}
  {Proceedings of the 54th Annual ACM SIGACT Symposium on Theory of
  Computing}}}\ (\bibinfo {year} {2022})\ pp.\ \bibinfo {pages}
  {375--388}\BibitemShut {NoStop}%
\bibitem [{\citenamefont {Dinur}\ \emph {et~al.}(2022)\citenamefont {Dinur},
  \citenamefont {Hsieh}, \citenamefont {Lin},\ and\ \citenamefont
  {Vidick}}]{dinur2022good}%
  \BibitemOpen
  \bibfield  {author} {\bibinfo {author} {\bibfnamefont {I.}~\bibnamefont
  {Dinur}}, \bibinfo {author} {\bibfnamefont {M.-H.}\ \bibnamefont {Hsieh}},
  \bibinfo {author} {\bibfnamefont {T.-C.}\ \bibnamefont {Lin}},\ and\ \bibinfo
  {author} {\bibfnamefont {T.}~\bibnamefont {Vidick}},\ }\bibfield  {title}
  {\bibinfo {title} {Good quantum ldpc codes with linear time decoders},\
  }\href {https://arxiv.org/abs/2206.07750} {\bibfield  {journal} {\bibinfo
  {journal} {arXiv preprint arXiv:2206.07750}\ } (\bibinfo {year}
  {2022})}\BibitemShut {NoStop}%
\bibitem [{\citenamefont {Breuckmann}\ and\ \citenamefont
  {Eberhardt}(2021{\natexlab{a}})}]{DBLP:journals/tit/BreuckmannE21}%
  \BibitemOpen
  \bibfield  {author} {\bibinfo {author} {\bibfnamefont {N.~P.}\ \bibnamefont
  {Breuckmann}}\ and\ \bibinfo {author} {\bibfnamefont {J.~N.}\ \bibnamefont
  {Eberhardt}},\ }\bibfield  {title} {\bibinfo {title} {Balanced product
  quantum codes},\ }\href {https://ieeexplore.ieee.org/document/9490244}
  {\bibfield  {journal} {\bibinfo  {journal} {IEEE Transactions on Information
  Theory}\ }\textbf {\bibinfo {volume} {67}},\ \bibinfo {pages} {6653}
  (\bibinfo {year} {2021}{\natexlab{a}})}\BibitemShut {NoStop}%
\bibitem [{\citenamefont {Leverrier}\ and\ \citenamefont
  {Z{\'e}mor}(2022)}]{leverrier2022quantum}%
  \BibitemOpen
  \bibfield  {author} {\bibinfo {author} {\bibfnamefont {A.}~\bibnamefont
  {Leverrier}}\ and\ \bibinfo {author} {\bibfnamefont {G.}~\bibnamefont
  {Z{\'e}mor}},\ }\bibfield  {title} {\bibinfo {title} {Quantum tanner codes},\
  }\href {https://arxiv.org/abs/2202.13641} {\bibfield  {journal} {\bibinfo
  {journal} {arXiv preprint arXiv:2202.13641}\ } (\bibinfo {year}
  {2022})}\BibitemShut {NoStop}%
\bibitem [{\citenamefont {Lin}\ and\ \citenamefont
  {Hsieh}(2022)}]{lin2022good}%
  \BibitemOpen
  \bibfield  {author} {\bibinfo {author} {\bibfnamefont {T.-C.}\ \bibnamefont
  {Lin}}\ and\ \bibinfo {author} {\bibfnamefont {M.-H.}\ \bibnamefont
  {Hsieh}},\ }\bibfield  {title} {\bibinfo {title} {Good quantum ldpc codes
  with linear time decoder from lossless expanders},\ }\href
  {https://arxiv.org/abs/2203.03581} {\bibfield  {journal} {\bibinfo  {journal}
  {arXiv:2203.03581}\ } (\bibinfo {year} {2022})}\BibitemShut {NoStop}%
\bibitem [{\citenamefont {Tremblay}\ \emph {et~al.}(2022)\citenamefont
  {Tremblay}, \citenamefont {Delfosse},\ and\ \citenamefont
  {Beverland}}]{tremblay2022constant}%
  \BibitemOpen
  \bibfield  {author} {\bibinfo {author} {\bibfnamefont {M.~A.}\ \bibnamefont
  {Tremblay}}, \bibinfo {author} {\bibfnamefont {N.}~\bibnamefont {Delfosse}},\
  and\ \bibinfo {author} {\bibfnamefont {M.~E.}\ \bibnamefont {Beverland}},\
  }\bibfield  {title} {\bibinfo {title} {Constant-overhead quantum error
  correction with thin planar connectivity},\ }\href
  {https://link.aps.org/doi/10.1103/PhysRevLett.129.050504} {\bibfield
  {journal} {\bibinfo  {journal} {Physical Review Letters}\ }\textbf {\bibinfo
  {volume} {129}},\ \bibinfo {pages} {050504} (\bibinfo {year}
  {2022})}\BibitemShut {NoStop}%
\bibitem [{\citenamefont {Baspin}\ and\ \citenamefont
  {Krishna}(2022{\natexlab{a}})}]{baspin2022connectivity}%
  \BibitemOpen
  \bibfield  {author} {\bibinfo {author} {\bibfnamefont {N.}~\bibnamefont
  {Baspin}}\ and\ \bibinfo {author} {\bibfnamefont {A.}~\bibnamefont
  {Krishna}},\ }\bibfield  {title} {\bibinfo {title} {Connectivity constrains
  quantum codes},\ }\href {https://doi.org/10.22331/q-2022-05-13-711}
  {\bibfield  {journal} {\bibinfo  {journal} {Quantum}\ }\textbf {\bibinfo
  {volume} {6}},\ \bibinfo {pages} {711} (\bibinfo {year}
  {2022}{\natexlab{a}})}\BibitemShut {NoStop}%
\bibitem [{\citenamefont {Baspin}\ and\ \citenamefont
  {Krishna}(2022{\natexlab{b}})}]{Baspin2022quantifying}%
  \BibitemOpen
  \bibfield  {author} {\bibinfo {author} {\bibfnamefont {N.}~\bibnamefont
  {Baspin}}\ and\ \bibinfo {author} {\bibfnamefont {A.}~\bibnamefont
  {Krishna}},\ }\bibfield  {title} {\bibinfo {title} {Quantifying nonlocality:
  How outperforming local quantum codes is expensive},\ }\href
  {https://doi.org/10.1103/PhysRevLett.129.050505} {\bibfield  {journal}
  {\bibinfo  {journal} {Phys. Rev. Lett.}\ }\textbf {\bibinfo {volume} {129}},\
  \bibinfo {pages} {050505} (\bibinfo {year} {2022}{\natexlab{b}})}\BibitemShut
  {NoStop}%
\bibitem [{\citenamefont {Gambetta}\ \emph {et~al.}(2020)\citenamefont
  {Gambetta}, \citenamefont {Zhang}, \citenamefont {Hennrich}, \citenamefont
  {Lesanovsky},\ and\ \citenamefont {Li}}]{Gambetta2020Rydberg}%
  \BibitemOpen
  \bibfield  {author} {\bibinfo {author} {\bibfnamefont {F.~M.}\ \bibnamefont
  {Gambetta}}, \bibinfo {author} {\bibfnamefont {C.}~\bibnamefont {Zhang}},
  \bibinfo {author} {\bibfnamefont {M.}~\bibnamefont {Hennrich}}, \bibinfo
  {author} {\bibfnamefont {I.}~\bibnamefont {Lesanovsky}},\ and\ \bibinfo
  {author} {\bibfnamefont {W.}~\bibnamefont {Li}},\ }\bibfield  {title}
  {\bibinfo {title} {Long-range multibody interactions and three-body
  antiblockade in a trapped rydberg ion chain},\ }\href
  {https://doi.org/10.1103/PhysRevLett.125.133602} {\bibfield  {journal}
  {\bibinfo  {journal} {Phys. Rev. Lett.}\ }\textbf {\bibinfo {volume} {125}},\
  \bibinfo {pages} {133602} (\bibinfo {year} {2020})}\BibitemShut {NoStop}%
\bibitem [{\citenamefont {Landsman}\ \emph {et~al.}(2019)\citenamefont
  {Landsman}, \citenamefont {Wu}, \citenamefont {Leung}, \citenamefont {Zhu},
  \citenamefont {Linke}, \citenamefont {Brown}, \citenamefont {Duan},\ and\
  \citenamefont {Monroe}}]{Landsman2019ion}%
  \BibitemOpen
  \bibfield  {author} {\bibinfo {author} {\bibfnamefont {K.~A.}\ \bibnamefont
  {Landsman}}, \bibinfo {author} {\bibfnamefont {Y.}~\bibnamefont {Wu}},
  \bibinfo {author} {\bibfnamefont {P.~H.}\ \bibnamefont {Leung}}, \bibinfo
  {author} {\bibfnamefont {D.}~\bibnamefont {Zhu}}, \bibinfo {author}
  {\bibfnamefont {N.~M.}\ \bibnamefont {Linke}}, \bibinfo {author}
  {\bibfnamefont {K.~R.}\ \bibnamefont {Brown}}, \bibinfo {author}
  {\bibfnamefont {L.}~\bibnamefont {Duan}},\ and\ \bibinfo {author}
  {\bibfnamefont {C.}~\bibnamefont {Monroe}},\ }\bibfield  {title} {\bibinfo
  {title} {Two-qubit entangling gates within arbitrarily long chains of trapped
  ions},\ }\href {https://link.aps.org/doi/10.1103/PhysRevA.100.022332}
  {\bibfield  {journal} {\bibinfo  {journal} {Physical Review A}\ }\textbf
  {\bibinfo {volume} {100}},\ \bibinfo {pages} {022332} (\bibinfo {year}
  {2019})}\BibitemShut {NoStop}%
\bibitem [{\citenamefont {H{\"a}ffner}\ \emph {et~al.}(2008)\citenamefont
  {H{\"a}ffner}, \citenamefont {Roos},\ and\ \citenamefont
  {Blatt}}]{HAFFNER2008}%
  \BibitemOpen
  \bibfield  {author} {\bibinfo {author} {\bibfnamefont {H.}~\bibnamefont
  {H{\"a}ffner}}, \bibinfo {author} {\bibfnamefont {C.~F.}\ \bibnamefont
  {Roos}},\ and\ \bibinfo {author} {\bibfnamefont {R.}~\bibnamefont {Blatt}},\
  }\bibfield  {title} {\bibinfo {title} {Quantum computing with trapped ions},\
  }\href {https://doi.org/10.1016/j.physrep.2008.09.003} {\bibfield  {journal}
  {\bibinfo  {journal} {Physics reports}\ }\textbf {\bibinfo {volume} {469}},\
  \bibinfo {pages} {155} (\bibinfo {year} {2008})}\BibitemShut {NoStop}%
\bibitem [{\citenamefont {Choe}\ and\ \citenamefont
  {Koenig}(2022)}]{Choe2022kam}%
  \BibitemOpen
  \bibfield  {author} {\bibinfo {author} {\bibfnamefont {S.~H.}\ \bibnamefont
  {Choe}}\ and\ \bibinfo {author} {\bibfnamefont {R.}~\bibnamefont {Koenig}},\
  }\bibfield  {title} {\bibinfo {title} {Long-range data transmission in a
  fault-tolerant quantum bus architecture},\ }\href
  {https://arxiv.org/abs/2203.03581} {\bibfield  {journal} {\bibinfo  {journal}
  {arXiv preprint arXiv:2209.09774}\ } (\bibinfo {year} {2022})}\BibitemShut
  {NoStop}%
\bibitem [{\citenamefont {Dennis}\ \emph {et~al.}(2002)\citenamefont {Dennis},
  \citenamefont {Kitaev}, \citenamefont {Landahl},\ and\ \citenamefont
  {Preskill}}]{Dennis2002Topo}%
  \BibitemOpen
  \bibfield  {author} {\bibinfo {author} {\bibfnamefont {E.}~\bibnamefont
  {Dennis}}, \bibinfo {author} {\bibfnamefont {A.}~\bibnamefont {Kitaev}},
  \bibinfo {author} {\bibfnamefont {A.}~\bibnamefont {Landahl}},\ and\ \bibinfo
  {author} {\bibfnamefont {J.}~\bibnamefont {Preskill}},\ }\bibfield  {title}
  {\bibinfo {title} {Topological quantum memory},\ }\href
  {https://doi.org/10.1063/1.1499754} {\bibfield  {journal} {\bibinfo
  {journal} {Journal of Mathematical Physics}\ }\textbf {\bibinfo {volume}
  {43}},\ \bibinfo {pages} {4452} (\bibinfo {year} {2002})}\BibitemShut
  {NoStop}%
\bibitem [{\citenamefont {Kitaev}(2003)}]{KITAEV2003}%
  \BibitemOpen
  \bibfield  {author} {\bibinfo {author} {\bibfnamefont {A.~Y.}\ \bibnamefont
  {Kitaev}},\ }\bibfield  {title} {\bibinfo {title} {Fault-tolerant quantum
  computation by anyons},\ }\href
  {https://doi.org/10.1016/S0003-4916(02)00018-0} {\bibfield  {journal}
  {\bibinfo  {journal} {Annals of Physics}\ }\textbf {\bibinfo {volume}
  {303}},\ \bibinfo {pages} {2} (\bibinfo {year} {2003})}\BibitemShut {NoStop}%
\bibitem [{\citenamefont {Bravyi}\ and\ \citenamefont
  {Kitaev}(1998)}]{bravyi1998quantum}%
  \BibitemOpen
  \bibfield  {author} {\bibinfo {author} {\bibfnamefont {S.~B.}\ \bibnamefont
  {Bravyi}}\ and\ \bibinfo {author} {\bibfnamefont {A.~Y.}\ \bibnamefont
  {Kitaev}},\ }\bibfield  {title} {\bibinfo {title} {Quantum codes on a lattice
  with boundary},\ }\href {https://arxiv.org/abs/quant-ph/9811052} {\bibfield
  {journal} {\bibinfo  {journal} {arXiv preprint quant-ph/9811052}\ } (\bibinfo
  {year} {1998})}\BibitemShut {NoStop}%
\bibitem [{\citenamefont {Krinner}\ \emph {et~al.}(2022)\citenamefont
  {Krinner}, \citenamefont {Lacroix}, \citenamefont {Remm}, \citenamefont
  {Di~Paolo}, \citenamefont {Genois}, \citenamefont {Leroux}, \citenamefont
  {Hellings}, \citenamefont {Lazar}, \citenamefont {Swiadek}, \citenamefont
  {Herrmann} \emph {et~al.}}]{krinner2022realizing}%
  \BibitemOpen
  \bibfield  {author} {\bibinfo {author} {\bibfnamefont {S.}~\bibnamefont
  {Krinner}}, \bibinfo {author} {\bibfnamefont {N.}~\bibnamefont {Lacroix}},
  \bibinfo {author} {\bibfnamefont {A.}~\bibnamefont {Remm}}, \bibinfo {author}
  {\bibfnamefont {A.}~\bibnamefont {Di~Paolo}}, \bibinfo {author}
  {\bibfnamefont {E.}~\bibnamefont {Genois}}, \bibinfo {author} {\bibfnamefont
  {C.}~\bibnamefont {Leroux}}, \bibinfo {author} {\bibfnamefont
  {C.}~\bibnamefont {Hellings}}, \bibinfo {author} {\bibfnamefont
  {S.}~\bibnamefont {Lazar}}, \bibinfo {author} {\bibfnamefont
  {F.}~\bibnamefont {Swiadek}}, \bibinfo {author} {\bibfnamefont
  {J.}~\bibnamefont {Herrmann}}, \emph {et~al.},\ }\bibfield  {title} {\bibinfo
  {title} {Realizing repeated quantum error correction in a distance-three
  surface code},\ }\href {https://doi.org/10.1038/s41586-022-04566-8}
  {\bibfield  {journal} {\bibinfo  {journal} {Nature}\ }\textbf {\bibinfo
  {volume} {605}},\ \bibinfo {pages} {669} (\bibinfo {year}
  {2022})}\BibitemShut {NoStop}%
\bibitem [{\citenamefont {Andersen}\ \emph {et~al.}(2020)\citenamefont
  {Andersen}, \citenamefont {Remm}, \citenamefont {Lazar}, \citenamefont
  {Krinner}, \citenamefont {Lacroix}, \citenamefont {Norris}, \citenamefont
  {Gabureac}, \citenamefont {Eichler},\ and\ \citenamefont
  {Wallraff}}]{andersen2020repeated}%
  \BibitemOpen
  \bibfield  {author} {\bibinfo {author} {\bibfnamefont {C.~K.}\ \bibnamefont
  {Andersen}}, \bibinfo {author} {\bibfnamefont {A.}~\bibnamefont {Remm}},
  \bibinfo {author} {\bibfnamefont {S.}~\bibnamefont {Lazar}}, \bibinfo
  {author} {\bibfnamefont {S.}~\bibnamefont {Krinner}}, \bibinfo {author}
  {\bibfnamefont {N.}~\bibnamefont {Lacroix}}, \bibinfo {author} {\bibfnamefont
  {G.~J.}\ \bibnamefont {Norris}}, \bibinfo {author} {\bibfnamefont
  {M.}~\bibnamefont {Gabureac}}, \bibinfo {author} {\bibfnamefont
  {C.}~\bibnamefont {Eichler}},\ and\ \bibinfo {author} {\bibfnamefont
  {A.}~\bibnamefont {Wallraff}},\ }\bibfield  {title} {\bibinfo {title}
  {Repeated quantum error detection in a surface code},\ }\href
  {https://doi.org/10.1038/s41567-020-0920-y} {\bibfield  {journal} {\bibinfo
  {journal} {Nature Physics}\ }\textbf {\bibinfo {volume} {16}},\ \bibinfo
  {pages} {875} (\bibinfo {year} {2020})}\BibitemShut {NoStop}%
\bibitem [{\citenamefont {Marques}\ \emph {et~al.}(2022)\citenamefont
  {Marques}, \citenamefont {Varbanov}, \citenamefont {Moreira}, \citenamefont
  {Ali}, \citenamefont {Muthusubramanian}, \citenamefont {Zachariadis},
  \citenamefont {Battistel}, \citenamefont {Beekman}, \citenamefont {Haider},
  \citenamefont {Vlothuizen} \emph {et~al.}}]{marques2022logical}%
  \BibitemOpen
  \bibfield  {author} {\bibinfo {author} {\bibfnamefont {J.}~\bibnamefont
  {Marques}}, \bibinfo {author} {\bibfnamefont {B.}~\bibnamefont {Varbanov}},
  \bibinfo {author} {\bibfnamefont {M.}~\bibnamefont {Moreira}}, \bibinfo
  {author} {\bibfnamefont {H.}~\bibnamefont {Ali}}, \bibinfo {author}
  {\bibfnamefont {N.}~\bibnamefont {Muthusubramanian}}, \bibinfo {author}
  {\bibfnamefont {C.}~\bibnamefont {Zachariadis}}, \bibinfo {author}
  {\bibfnamefont {F.}~\bibnamefont {Battistel}}, \bibinfo {author}
  {\bibfnamefont {M.}~\bibnamefont {Beekman}}, \bibinfo {author} {\bibfnamefont
  {N.}~\bibnamefont {Haider}}, \bibinfo {author} {\bibfnamefont
  {W.}~\bibnamefont {Vlothuizen}}, \emph {et~al.},\ }\bibfield  {title}
  {\bibinfo {title} {Logical-qubit operations in an error-detecting surface
  code},\ }\href {https://doi.org/10.1038/s41567-021-01423-9} {\bibfield
  {journal} {\bibinfo  {journal} {Nature Physics}\ }\textbf {\bibinfo {volume}
  {18}},\ \bibinfo {pages} {80} (\bibinfo {year} {2022})}\BibitemShut {NoStop}%
\bibitem [{\citenamefont {Chen}\ \emph {et~al.}(2021)\citenamefont {Chen},
  \citenamefont {Satzinger}, \citenamefont {Atalaya}, \citenamefont {Korotkov},
  \citenamefont {Dunsworth}, \citenamefont {Sank}, \citenamefont {Quintana},
  \citenamefont {McEwen}, \citenamefont {Barends}, \citenamefont {Klimov} \emph
  {et~al.}}]{chen2021exponential}%
  \BibitemOpen
  \bibfield  {author} {\bibinfo {author} {\bibfnamefont {Z.}~\bibnamefont
  {Chen}}, \bibinfo {author} {\bibfnamefont {K.~J.}\ \bibnamefont {Satzinger}},
  \bibinfo {author} {\bibfnamefont {J.}~\bibnamefont {Atalaya}}, \bibinfo
  {author} {\bibfnamefont {A.~N.}\ \bibnamefont {Korotkov}}, \bibinfo {author}
  {\bibfnamefont {A.}~\bibnamefont {Dunsworth}}, \bibinfo {author}
  {\bibfnamefont {D.}~\bibnamefont {Sank}}, \bibinfo {author} {\bibfnamefont
  {C.}~\bibnamefont {Quintana}}, \bibinfo {author} {\bibfnamefont
  {M.}~\bibnamefont {McEwen}}, \bibinfo {author} {\bibfnamefont
  {R.}~\bibnamefont {Barends}}, \bibinfo {author} {\bibfnamefont {P.~V.}\
  \bibnamefont {Klimov}}, \emph {et~al.},\ }\bibfield  {title} {\bibinfo
  {title} {Exponential suppression of bit or phase errors with cyclic error
  correction},\ }\href {https://doi.org/10.1038/s41586-021-03588-y} {\bibfield
  {journal} {\bibinfo  {journal} {Nature}\ }\textbf {\bibinfo {volume} {595}},\
  \bibinfo {pages} {383} (\bibinfo {year} {2021})}\BibitemShut {NoStop}%
\bibitem [{\citenamefont {Acharya}\ \emph {et~al.}(2022)\citenamefont
  {Acharya}, \citenamefont {Aleiner}, \citenamefont {Allen}, \citenamefont
  {Andersen}, \citenamefont {Ansmann}, \citenamefont {Arute}, \citenamefont
  {Arya}, \citenamefont {Asfaw}, \citenamefont {Atalaya}, \citenamefont
  {Babbush} \emph {et~al.}}]{acharya2022suppressing}%
  \BibitemOpen
  \bibfield  {author} {\bibinfo {author} {\bibfnamefont {R.}~\bibnamefont
  {Acharya}}, \bibinfo {author} {\bibfnamefont {I.}~\bibnamefont {Aleiner}},
  \bibinfo {author} {\bibfnamefont {R.}~\bibnamefont {Allen}}, \bibinfo
  {author} {\bibfnamefont {T.~I.}\ \bibnamefont {Andersen}}, \bibinfo {author}
  {\bibfnamefont {M.}~\bibnamefont {Ansmann}}, \bibinfo {author} {\bibfnamefont
  {F.}~\bibnamefont {Arute}}, \bibinfo {author} {\bibfnamefont
  {K.}~\bibnamefont {Arya}}, \bibinfo {author} {\bibfnamefont {A.}~\bibnamefont
  {Asfaw}}, \bibinfo {author} {\bibfnamefont {J.}~\bibnamefont {Atalaya}},
  \bibinfo {author} {\bibfnamefont {R.}~\bibnamefont {Babbush}}, \emph
  {et~al.},\ }\bibfield  {title} {\bibinfo {title} {Suppressing quantum errors
  by scaling a surface code logical qubit},\ }\href
  {https://arxiv.org/abs/2207.06431} {\bibfield  {journal} {\bibinfo  {journal}
  {arXiv preprint arXiv:2207.06431}\ } (\bibinfo {year} {2022})}\BibitemShut
  {NoStop}%
\bibitem [{\citenamefont {Campbell}\ \emph {et~al.}(2017)\citenamefont
  {Campbell}, \citenamefont {Terhal},\ and\ \citenamefont
  {Vuillot}}]{campbell2017roads}%
  \BibitemOpen
  \bibfield  {author} {\bibinfo {author} {\bibfnamefont {E.~T.}\ \bibnamefont
  {Campbell}}, \bibinfo {author} {\bibfnamefont {B.~M.}\ \bibnamefont
  {Terhal}},\ and\ \bibinfo {author} {\bibfnamefont {C.}~\bibnamefont
  {Vuillot}},\ }\bibfield  {title} {\bibinfo {title} {Roads towards
  fault-tolerant universal quantum computation},\ }\href
  {https://doi.org/10.1038/nature23460} {\bibfield  {journal} {\bibinfo
  {journal} {Nature}\ }\textbf {\bibinfo {volume} {549}},\ \bibinfo {pages}
  {172} (\bibinfo {year} {2017})}\BibitemShut {NoStop}%
\bibitem [{\citenamefont {Nickerson}\ \emph {et~al.}(2013)\citenamefont
  {Nickerson}, \citenamefont {Li},\ and\ \citenamefont
  {Benjamin}}]{nickerson2013}%
  \BibitemOpen
  \bibfield  {author} {\bibinfo {author} {\bibfnamefont {N.~H.}\ \bibnamefont
  {Nickerson}}, \bibinfo {author} {\bibfnamefont {Y.}~\bibnamefont {Li}},\ and\
  \bibinfo {author} {\bibfnamefont {S.~C.}\ \bibnamefont {Benjamin}},\
  }\bibfield  {title} {\bibinfo {title} {Topological quantum computing with a
  very noisy network and local error rates approaching one percent},\ }\href
  {https://doi.org/10.1038/ncomms2773} {\bibfield  {journal} {\bibinfo
  {journal} {Nature communications}\ }\textbf {\bibinfo {volume} {4}},\
  \bibinfo {pages} {1} (\bibinfo {year} {2013})}\BibitemShut {NoStop}%
\bibitem [{\citenamefont {Cai}\ \emph {et~al.}(2022)\citenamefont {Cai},
  \citenamefont {Siegel},\ and\ \citenamefont {Benjamin}}]{cai2022looped}%
  \BibitemOpen
  \bibfield  {author} {\bibinfo {author} {\bibfnamefont {Z.}~\bibnamefont
  {Cai}}, \bibinfo {author} {\bibfnamefont {A.}~\bibnamefont {Siegel}},\ and\
  \bibinfo {author} {\bibfnamefont {S.}~\bibnamefont {Benjamin}},\ }\bibfield
  {title} {\bibinfo {title} {Looped pipelines enabling effective 3d qubit
  lattices in a strictly 2d device},\ }\href {https://arxiv.org/abs/2203.13123}
  {\bibfield  {journal} {\bibinfo  {journal} {arXiv preprint arXiv:2203.13123}\
  } (\bibinfo {year} {2022})}\BibitemShut {NoStop}%
\bibitem [{\citenamefont {Breuckmann}\ and\ \citenamefont
  {Eberhardt}(2021{\natexlab{b}})}]{breuckmann2021quantum}%
  \BibitemOpen
  \bibfield  {author} {\bibinfo {author} {\bibfnamefont {N.~P.}\ \bibnamefont
  {Breuckmann}}\ and\ \bibinfo {author} {\bibfnamefont {J.~N.}\ \bibnamefont
  {Eberhardt}},\ }\bibfield  {title} {\bibinfo {title} {Quantum low-density
  parity-check codes},\ }\href
  {https://link.aps.org/doi/10.1103/PRXQuantum.2.040101} {\bibfield  {journal}
  {\bibinfo  {journal} {PRX Quantum}\ }\textbf {\bibinfo {volume} {2}},\
  \bibinfo {pages} {040101} (\bibinfo {year} {2021}{\natexlab{b}})}\BibitemShut
  {NoStop}%
\bibitem [{\citenamefont {Panteleev}\ and\ \citenamefont
  {Kalachev}(2021{\natexlab{a}})}]{panteleev2021degenerate}%
  \BibitemOpen
  \bibfield  {author} {\bibinfo {author} {\bibfnamefont {P.}~\bibnamefont
  {Panteleev}}\ and\ \bibinfo {author} {\bibfnamefont {G.}~\bibnamefont
  {Kalachev}},\ }\bibfield  {title} {\bibinfo {title} {Degenerate quantum ldpc
  codes with good finite length performance},\ }\href
  {https://doi.org/10.22331/q-2021-11-22-585} {\bibfield  {journal} {\bibinfo
  {journal} {Quantum}\ }\textbf {\bibinfo {volume} {5}},\ \bibinfo {pages}
  {585} (\bibinfo {year} {2021}{\natexlab{a}})}\BibitemShut {NoStop}%
\bibitem [{\citenamefont {Bombin}\ and\ \citenamefont
  {Martin-Delgado}(2006)}]{bombin2006topological}%
  \BibitemOpen
  \bibfield  {author} {\bibinfo {author} {\bibfnamefont {H.}~\bibnamefont
  {Bombin}}\ and\ \bibinfo {author} {\bibfnamefont {M.}~\bibnamefont
  {Martin-Delgado}},\ }\bibfield  {title} {\bibinfo {title} {Topological
  quantum error correction with optimal encoding rate},\ }\href
  {https://link.aps.org/doi/10.1103/PhysRevA.73.062303} {\bibfield  {journal}
  {\bibinfo  {journal} {Physical Review A}\ }\textbf {\bibinfo {volume} {73}},\
  \bibinfo {pages} {062303} (\bibinfo {year} {2006})}\BibitemShut {NoStop}%
\bibitem [{\citenamefont {Freedman}\ \emph {et~al.}(2002)\citenamefont
  {Freedman}, \citenamefont {Meyer},\ and\ \citenamefont
  {Luo}}]{freedman2002z2}%
  \BibitemOpen
  \bibfield  {author} {\bibinfo {author} {\bibfnamefont {M.~H.}\ \bibnamefont
  {Freedman}}, \bibinfo {author} {\bibfnamefont {D.~A.}\ \bibnamefont
  {Meyer}},\ and\ \bibinfo {author} {\bibfnamefont {F.}~\bibnamefont {Luo}},\
  }\bibfield  {title} {\bibinfo {title} {Z2-systolic freedom and quantum
  codes},\ }in\ \href@noop {} {\emph {\bibinfo {booktitle} {Mathematics of
  quantum computation}}}\ (\bibinfo  {publisher} {Chapman and Hall/CRC},\
  \bibinfo {year} {2002})\ pp.\ \bibinfo {pages} {303--338}\BibitemShut
  {NoStop}%
\bibitem [{\citenamefont {Gottesman}(1997)}]{gottesman1997stabilizer}%
  \BibitemOpen
  \bibfield  {author} {\bibinfo {author} {\bibfnamefont {D.}~\bibnamefont
  {Gottesman}},\ }\href@noop {} {\emph {\bibinfo {title} {Stabilizer codes and
  quantum error correction}}}\ (\bibinfo  {publisher} {California Institute of
  Technology},\ \bibinfo {year} {1997})\BibitemShut {NoStop}%
\bibitem [{\citenamefont {Fogarty}(4 01)}]{privcorrfogarty}%
  \BibitemOpen
  \bibfield  {author} {\bibinfo {author} {\bibfnamefont {M.}~\bibnamefont
  {Fogarty}},\ }\href@noop {} {}\bibinfo {howpublished} {personal
  communication} (\bibinfo {year} {2022-04-01})\BibitemShut {NoStop}%
\bibitem [{\citenamefont {Delfosse}\ \emph {et~al.}(2021)\citenamefont
  {Delfosse}, \citenamefont {Beverland},\ and\ \citenamefont
  {Tremblay}}]{delfosse2021bounds}%
  \BibitemOpen
  \bibfield  {author} {\bibinfo {author} {\bibfnamefont {N.}~\bibnamefont
  {Delfosse}}, \bibinfo {author} {\bibfnamefont {M.~E.}\ \bibnamefont
  {Beverland}},\ and\ \bibinfo {author} {\bibfnamefont {M.~A.}\ \bibnamefont
  {Tremblay}},\ }\bibfield  {title} {\bibinfo {title} {Bounds on stabilizer
  measurement circuits and obstructions to local implementations of quantum
  ldpc codes},\ }\href {https://arxiv.org/abs/2109.14599} {\bibfield  {journal}
  {\bibinfo  {journal} {arXiv preprint arXiv:2109.14599}\ } (\bibinfo {year}
  {2021})}\BibitemShut {NoStop}%
\bibitem [{\citenamefont {Pattison}\ \emph {et~al.}(2023)\citenamefont
  {Pattison}, \citenamefont {Krishna},\ and\ \citenamefont
  {Preskill}}]{pattison2023hierarchical}%
  \BibitemOpen
  \bibfield  {author} {\bibinfo {author} {\bibfnamefont {C.~A.}\ \bibnamefont
  {Pattison}}, \bibinfo {author} {\bibfnamefont {A.}~\bibnamefont {Krishna}},\
  and\ \bibinfo {author} {\bibfnamefont {J.}~\bibnamefont {Preskill}},\
  }\bibfield  {title} {\bibinfo {title} {Hierarchical memories: Simulating
  quantum ldpc codes with local gates},\ }\href
  {https://arxiv.org/abs/2303.04798} {\bibfield  {journal} {\bibinfo  {journal}
  {arXiv preprint arXiv:2303.04798}\ } (\bibinfo {year} {2023})}\BibitemShut
  {NoStop}%
\bibitem [{\citenamefont {Tillich}\ and\ \citenamefont
  {Z{\'e}mor}(2013)}]{tillich2013quantum}%
  \BibitemOpen
  \bibfield  {author} {\bibinfo {author} {\bibfnamefont {J.-P.}\ \bibnamefont
  {Tillich}}\ and\ \bibinfo {author} {\bibfnamefont {G.}~\bibnamefont
  {Z{\'e}mor}},\ }\bibfield  {title} {\bibinfo {title} {Quantum ldpc codes with
  positive rate and minimum distance proportional to the square root of the
  blocklength},\ }\href {https://ieeexplore.ieee.org/document/6671468}
  {\bibfield  {journal} {\bibinfo  {journal} {IEEE Transactions on Information
  Theory}\ }\textbf {\bibinfo {volume} {60}},\ \bibinfo {pages} {1193}
  (\bibinfo {year} {2013})}\BibitemShut {NoStop}%
\bibitem [{\citenamefont {Zeng}\ and\ \citenamefont
  {Pryadko}(2019)}]{zeng2019higher}%
  \BibitemOpen
  \bibfield  {author} {\bibinfo {author} {\bibfnamefont {W.}~\bibnamefont
  {Zeng}}\ and\ \bibinfo {author} {\bibfnamefont {L.~P.}\ \bibnamefont
  {Pryadko}},\ }\bibfield  {title} {\bibinfo {title} {Higher-dimensional
  quantum hypergraph-product codes with finite rates},\ }\href
  {https://link.aps.org/doi/10.1103/PhysRevLett.122.230501} {\bibfield
  {journal} {\bibinfo  {journal} {Physical review letters}\ }\textbf {\bibinfo
  {volume} {122}},\ \bibinfo {pages} {230501} (\bibinfo {year}
  {2019})}\BibitemShut {NoStop}%
\bibitem [{\citenamefont {Higgott}\ and\ \citenamefont
  {Breuckmann}(2022)}]{higgott2022improved}%
  \BibitemOpen
  \bibfield  {author} {\bibinfo {author} {\bibfnamefont {O.}~\bibnamefont
  {Higgott}}\ and\ \bibinfo {author} {\bibfnamefont {N.~P.}\ \bibnamefont
  {Breuckmann}},\ }\bibfield  {title} {\bibinfo {title} {Improved single-shot
  decoding of higher dimensional hypergraph product codes},\ }\href
  {https://arxiv.org/abs/2206.03122} {\bibfield  {journal} {\bibinfo  {journal}
  {arXiv preprint arXiv:2206.03122}\ } (\bibinfo {year} {2022})}\BibitemShut
  {NoStop}%
\bibitem [{git(2022)}]{githurepo2}%
  \BibitemOpen
  \href@noop {} {}\bibinfo {howpublished}
  {\url{https://github.com/PurePhys/QuantumPCMs}} (\bibinfo {year}
  {2022})\BibitemShut {NoStop}%
\bibitem [{\citenamefont {Panteleev}\ and\ \citenamefont
  {Kalachev}(2021{\natexlab{b}})}]{panteleev2021quantum}%
  \BibitemOpen
  \bibfield  {author} {\bibinfo {author} {\bibfnamefont {P.}~\bibnamefont
  {Panteleev}}\ and\ \bibinfo {author} {\bibfnamefont {G.}~\bibnamefont
  {Kalachev}},\ }\bibfield  {title} {\bibinfo {title} {Quantum ldpc codes with
  almost linear minimum distance},\ }\href
  {https://ieeexplore.ieee.org/document/9567703} {\bibfield  {journal}
  {\bibinfo  {journal} {IEEE Transactions on Information Theory}\ }\textbf
  {\bibinfo {volume} {68}},\ \bibinfo {pages} {213} (\bibinfo {year}
  {2021}{\natexlab{b}})}\BibitemShut {NoStop}%
\bibitem [{\citenamefont {Hastings}\ \emph {et~al.}(2021)\citenamefont
  {Hastings}, \citenamefont {Haah},\ and\ \citenamefont
  {O'Donnell}}]{hastings2021fiber}%
  \BibitemOpen
  \bibfield  {author} {\bibinfo {author} {\bibfnamefont {M.~B.}\ \bibnamefont
  {Hastings}}, \bibinfo {author} {\bibfnamefont {J.}~\bibnamefont {Haah}},\
  and\ \bibinfo {author} {\bibfnamefont {R.}~\bibnamefont {O'Donnell}},\
  }\bibfield  {title} {\bibinfo {title} {Fiber bundle codes: breaking the
  $n^{1/2}$polylog($n$) barrier for quantum ldpc codes},\ }in\ \href
  {https://doi.org/10.1145/3406325.3451005} {\emph {\bibinfo {booktitle}
  {Proceedings of the 53rd Annual ACM SIGACT Symposium on Theory of
  Computing}}}\ (\bibinfo {year} {2021})\ pp.\ \bibinfo {pages}
  {1276--1288}\BibitemShut {NoStop}%
\bibitem [{\citenamefont {Campbell}(2007)}]{campbell2007}%
  \BibitemOpen
  \bibfield  {author} {\bibinfo {author} {\bibfnamefont {E.~T.}\ \bibnamefont
  {Campbell}},\ }\bibfield  {title} {\bibinfo {title} {Distributed
  quantum-information processing with minimal local resources},\ }\href
  {https://doi.org/10.1103/PhysRevA.76.040302} {\bibfield  {journal} {\bibinfo
  {journal} {Phys. Rev. A}\ }\textbf {\bibinfo {volume} {76}},\ \bibinfo
  {pages} {040302} (\bibinfo {year} {2007})}\BibitemShut {NoStop}%
\bibitem [{\citenamefont {Stephenson}\ \emph {et~al.}(2020)\citenamefont
  {Stephenson}, \citenamefont {Nadlinger}, \citenamefont {Nichol},
  \citenamefont {An}, \citenamefont {Drmota}, \citenamefont {Ballance},
  \citenamefont {Thirumalai}, \citenamefont {Goodwin}, \citenamefont {Lucas},\
  and\ \citenamefont {Ballance}}]{Stephenson2020}%
  \BibitemOpen
  \bibfield  {author} {\bibinfo {author} {\bibfnamefont {L.~J.}\ \bibnamefont
  {Stephenson}}, \bibinfo {author} {\bibfnamefont {D.~P.}\ \bibnamefont
  {Nadlinger}}, \bibinfo {author} {\bibfnamefont {B.~C.}\ \bibnamefont
  {Nichol}}, \bibinfo {author} {\bibfnamefont {S.}~\bibnamefont {An}}, \bibinfo
  {author} {\bibfnamefont {P.}~\bibnamefont {Drmota}}, \bibinfo {author}
  {\bibfnamefont {T.~G.}\ \bibnamefont {Ballance}}, \bibinfo {author}
  {\bibfnamefont {K.}~\bibnamefont {Thirumalai}}, \bibinfo {author}
  {\bibfnamefont {J.~F.}\ \bibnamefont {Goodwin}}, \bibinfo {author}
  {\bibfnamefont {D.~M.}\ \bibnamefont {Lucas}},\ and\ \bibinfo {author}
  {\bibfnamefont {C.~J.}\ \bibnamefont {Ballance}},\ }\bibfield  {title}
  {\bibinfo {title} {High-rate, high-fidelity entanglement of qubits across an
  elementary quantum network},\ }\href
  {https://doi.org/10.1103/PhysRevLett.124.110501} {\bibfield  {journal}
  {\bibinfo  {journal} {Phys. Rev. Lett.}\ }\textbf {\bibinfo {volume} {124}},\
  \bibinfo {pages} {110501} (\bibinfo {year} {2020})}\BibitemShut {NoStop}%
\bibitem [{\citenamefont {Pryadko}\ \emph {et~al.}(2022)\citenamefont
  {Pryadko}, \citenamefont {Shabashov},\ and\ \citenamefont
  {Kozin}}]{QDistRnd0.8.5}%
  \BibitemOpen
  \bibfield  {author} {\bibinfo {author} {\bibfnamefont {L.~P.}\ \bibnamefont
  {Pryadko}}, \bibinfo {author} {\bibfnamefont {V.~A.}\ \bibnamefont
  {Shabashov}},\ and\ \bibinfo {author} {\bibfnamefont {V.~K.}\ \bibnamefont
  {Kozin}},\ }\bibfield  {title} {\bibinfo {title} {Qdistrnd: A gap package for
  computing the distance of quantum error-correcting codes},\ }\href
  {https://doi.org/10.21105/joss.04120} {\bibfield  {journal} {\bibinfo
  {journal} {Journal of Open Source Software}\ }\textbf {\bibinfo {volume}
  {7}},\ \bibinfo {pages} {4120} (\bibinfo {year} {2022})}\BibitemShut
  {NoStop}%
\end{thebibliography}%

\end{document}